\documentclass[journal]{IEEEtran}
\usepackage{cite}

\ifCLASSINFOpdf
   \usepackage[pdftex]{graphicx}
\else
\fi
\usepackage{amsmath}
\usepackage{array}

\usepackage{pdfsync}

\usepackage{balance}

\usepackage{tabulary,ragged2e} 
\newcolumntype{Z}{>{\Centering\arraybackslash}L} 
\usepackage{listings}
\usepackage{multirow}
\newcolumntype{P}[1]{>{\centering\arraybackslash}p{#1}}
\newcolumntype{M}[1]{>{\centering\arraybackslash}m{#1}}
\usepackage{booktabs}
\usepackage{caption}
\usepackage[square,numbers]{natbib}
\usepackage[stable,bottom]{footmisc}
\usepackage{framed}
\usepackage{comment}
\usepackage{xcolor}
\usepackage{amssymb}

\usepackage{dblfloatfix}

\newcommand{\disablewr}[1]{#1}

\renewcommand{\disablewr}[1]{}%

\usepackage{authblk}

\usepackage{cleveref}
\usepackage{todonotes}
\usepackage{bytefield}
\usepackage{comment}

\usepackage{subcaption}
\usepackage{graphicx}

\usepackage{amsthm}
\usepackage{mathtools}
\usepackage{amsmath}
\usepackage{mathtools}

\DeclareMathOperator{\diag}{\mathrm{diag}\,}

\newcommand{\mt}[1]{\mathbf{#1}}
\DeclareMathOperator{\T}{^\mathsf{T}\,}
\DeclareMathOperator{\lbar}{\bar{\lambda}\,}

\newcommand\scalemath[2]{\scalebox{#1}{\mbox{\ensuremath{\displaystyle #2}}}}

\def\E{\mathbb{E}}
\def\P{\mathbb{P}}
\def\Nats{\mathbb{N}}

\newcommand{\mst}{\;\mathrm{ s.t. }}
\newcommand{\ind}[1]{\mathbf{1}_{#1}}

\newtheorem{theorem}{Theorem}[section]

\newtheorem{lemma}[theorem]{Lemma}

\begin{document}

%
\title{A Palm Calculus Approach to the Distribution of the Age of Information}
%
%
%

\author{
Amr Rizk,~\IEEEmembership{Senior Member,~IEEE,} and
Jean-Yves Le Boudec,~\IEEEmembership{Fellow,~IEEE}
\IEEEcompsocitemizethanks{
\IEEEcompsocthanksitem \hspace{-13pt} A. Rizk is with the University of Duisburg-Essen, Essen, Germany. (e-mail: amr.rizk@uni-due.de).\newline
\IEEEcompsocthanksitem Jean-Yves Le Boudec is with the École Polytechnique Fédérale de Lausanne (EPFL), 1015 Lausanne, Switzerland (e-mail: jean-yves.leboudec@epfl.ch).
}
}

%
%

%

\IEEEoverridecommandlockouts



\maketitle

\begin{abstract}
A key metric to express the timeliness of status updates in latency-sensitive networked systems is the age of information (AoI), i.e., the time elapsed since the generation of the last received informative status message.
This metric allows studying a number of applications including updates of sensory and control information in cyber-physical systems and vehicular networks as well as, job and resource allocation in cloud clusters.
State-of-the-art approaches to analyzing the AoI rely on queueing models that are composed of one or many queuing systems endowed with service order, e.g., FIFO, LIFO, or last-generated-first-out order.
A major difficulty arising in these analysis methods is capturing the AoI under message reordering when the delivery is non-preemptive and non-FIFO, i.e., when messages can overtake each other and the reception of informative messages may obsolete some messages that are underway.
In this paper, we derive an exact formulation for the distribution of AoI in non-preemptive, non-FIFO systems where the main ingredients of our analysis are Palm calculus and time inversion.
Owing to the rationality of the Laplace-Stieltjes transforms that are used in our approach, we obtain computable exact expressions for the distribution of AoI.

\end{abstract}


\section{Introduction}

Cyber-physical systems (CPS) constitute a type of hybrid system hat combines physical processes and computation~\cite{Lee08:CPS}.
Often, the considered physical process such as those arising in chemical plants or platoons of automated vehicles are controlled via feedback loops.
Due to this sensor-computation-actuator feedback loop, CPS are characterized by the mutual interaction of the physical process, the software computation, and essentially, the network.

While CPS encompass diverse key interactions worth accurate modeling such as control correctness and concurrency we are concerned in the following with the effort of characterizing the timeliness of sensor data when received at the controller.
This is the first step to ensure that the actions taken by the controller and hence executed by the actuator are based on fresh information.
A key metric to express this timeliness of sensor data at the controller is the Age of Information (AoI)~\cite{Kaul12}, which has recently been a vivid object of study~\cite{Yates:MDS:TINT20,Bedewy17}.
AoI is a semi-continuous function that denotes the age of the sensor (sender) status at the controller (receiver).
The status age is hence best described by a \emph{jump-and-drift} process that grows linearly with time and jumps downwards at the time points when informative messages arrive at the receiver.
An informative message is defined as a message containing an update that was generated after the generation time point of the last received update at the controller.
Now, the time points at which messages arrive at the receiver, as well as the timestamps contained in these messages, are random and essentially dependent on the generation and transmission of messages at the sender and at every network node on the path from the sender to the receiver.
Figure~\ref{fig:age_sample_path} shows a sketch of this scenario where a newer message\footnote{in terms of the generation time point} ($m_2$)
overtakes an older message ($m_1$).
Note that the age at the receiver does not jump downwards at the reception of the outdated message $m_1$.

One research direction to optimize the timeliness of sensor information in CPS is through advancing the state-of-the-art physical layer techniques such as deterministically reserved transmission time slots over all available frequencies as in low-latency 5G network slices known as URLLC~\cite{Popovski:19}.
While this eliminates contention on the wireless link in 5G, data packet interactions and sporadic network congestion still occur on the end-to-end path between the sensors and the controller.

Research on the topic of AoI has been characterized by the analysis of mathematical models that capture the stochastic process of the age at the receiver given a combination of ingredients, i.e., (i) the process of data generation and transmission at the source, (ii) a model of the network interactions such as traffic scheduling and the variability of the link transmission rate.
The prevalent approach in many works on AoI is to capture these ingredients in form of a queueing system (or a series thereof) that naturally capture the former and models the latter through the service process.
In many works the arrival process is often considered as a Poisson process for tractability~\cite{Kam13,Kam16:ToIT} or as a periodic process to capture simple sensor device implementations~\cite{Fidler21:AoI}.
The variety of queueing models ranges from simple M/M/1 queues with FIFO service to preemptive Last-Generated, First-Served (LGFS) systems.
A remarkable difficulty of some AoI models is due to the lack of FIFO service.
Allowing messages to overtake each other leads to considerable complexity as shown in the basic example in Figure~\ref{fig:age_sample_path}.
A  direct approach to model AoI systems with non-FIFO service is presented in \cite{Yates:multiple_source:TINT19} using the Stochastic-Hybrid-System (SHS) technique.
This technique essentially depends on the Fokker-Planck partial differential equation (PDE) satisfied by the time-dependent probability density of the AoI as shown in \cite{Yates:MDS:TINT20} and quickly becomes intractable.
For a comprehensive overview we point the reader to \cite{YatesSBKMU21a}.

The \textbf{key differences} of this work to the works in \cite{Kam16:ToIT,Yates18,YatesSBKMU21a} are: (i) Instead of considering the time variant PDE of the density of the age we reduce the problem to a simple model where we are interested in the stationary distribution of the age in a system where message overtaking is allowed. (ii) We obtain the distribution of the age using Palm calculus and time inversion where we essentially require to know the joint distribution of the age when a message arrives and the time until the next message arrives. Note that the obtained inversion formula applies to all types of arrivals in this model but we apply it to the informative messages to obtain the age density at the receiver. This model naturally captures the distribution of functions of the age. (iii) Due to the used mathematical tools our results only require \emph{stationarity} of the underlying queuing model, which is a Markov process on a discrete state space and can thus be analyzed with elementary techniques. (iv) The model considered in this paper is different from \cite{Kam16:ToIT,Yates18} as we consider a window flow controlled sender that injects at most a fixed number of non-obsolete messages into the network channel. We denote this model as $M/M/I_{\max}/I_{\max}^*$.

Our contributions in this paper are summarized as follows:
\begin{itemize}
  \item We use Palm calculus and time inversion to derive the probability distribution of the age of information in a stationary $M/M/I_{\max}/I_{\max}^*$  system.
  \item We calculate the Laplace-Stieltjes transform of the distribution of the age at the arrival time points of informative messages as well as at any point in time.
\end{itemize} 
\vspace{-10pt}
\section{Problem Statement and System Model}
\label{sec:system_model}
\label{ch:problem}
\begin{figure}
  \centering
  \includegraphics[width=\linewidth]{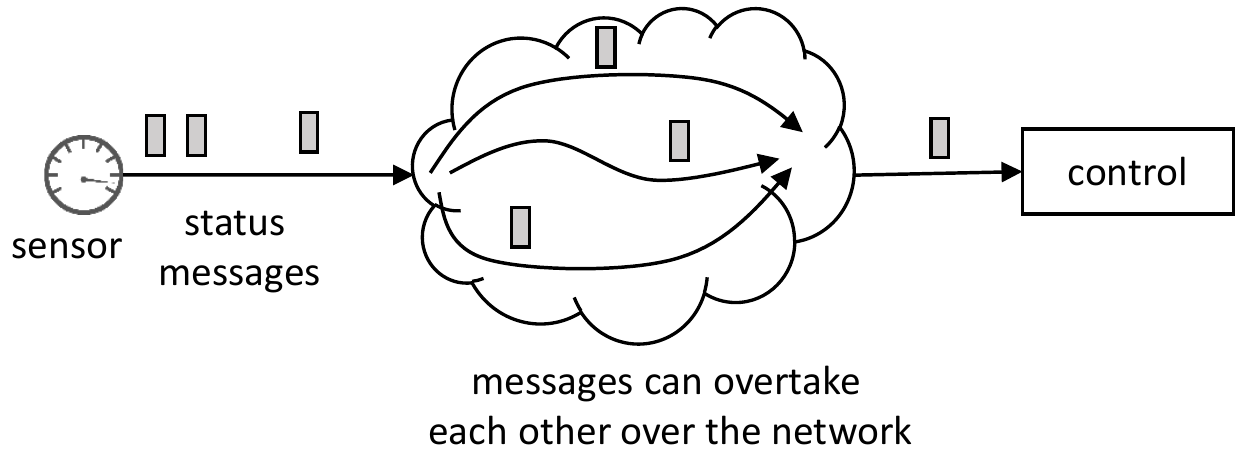}
  \caption{The sensory information is generated and immediately transmitted in form of messages. These can overtake each other on the network. Informative messages keep the total message order at the receiver and reduce the age of the status information at the receiver to their respective one way delay.}
  \label{fig:cps_scenario}
  \vspace{-20pt}
\end{figure}
We consider cyber-physical systems as depicted in Fig.~\ref{fig:cps_scenario} where sensors transmit status updates to a central control and data acquisition function.
We assume that timestamped messages are transmitted at the sender according to some parameterized stochastic process. When a message is generated, it obsoletes any previous message.
However, every message is subject to a one way delay and \emph{messages can overtake each other}. We say that a message is ``informative" if its timestamp was generated after the generation time of all messages received so far. When a non informative message arrives, it is of no use and is discarded.

We are interested in the age of information at the receiver, $X_t$, which is formally defined as follows. Timestamped messages are generated at times $\{\tau_i\}$ and received at times $\{\tau'_i\}$ respectively (with $\tau_i\leq \tau'_i$ ). Then
\begin{equation}
X_t= t -\max_{i: \tau'_i\leq t} \tau_i
\end{equation}
The dynamic evolution of $X_t$ is such that
the age $X_t$ increases at rate $1$ between arrival events; furthermore, when message $i$ arrives, the value of $X_t$ just after the arrival, namely $X_{\tau_i^+}$, is set to  $\min\left(\tau'_i-\tau_i\;, X_{\tau_i^-}\right)$ as seen in Figure~\ref{fig:age_sample_path}.


\begin{figure}
  \centering
  \includegraphics[width=\linewidth]{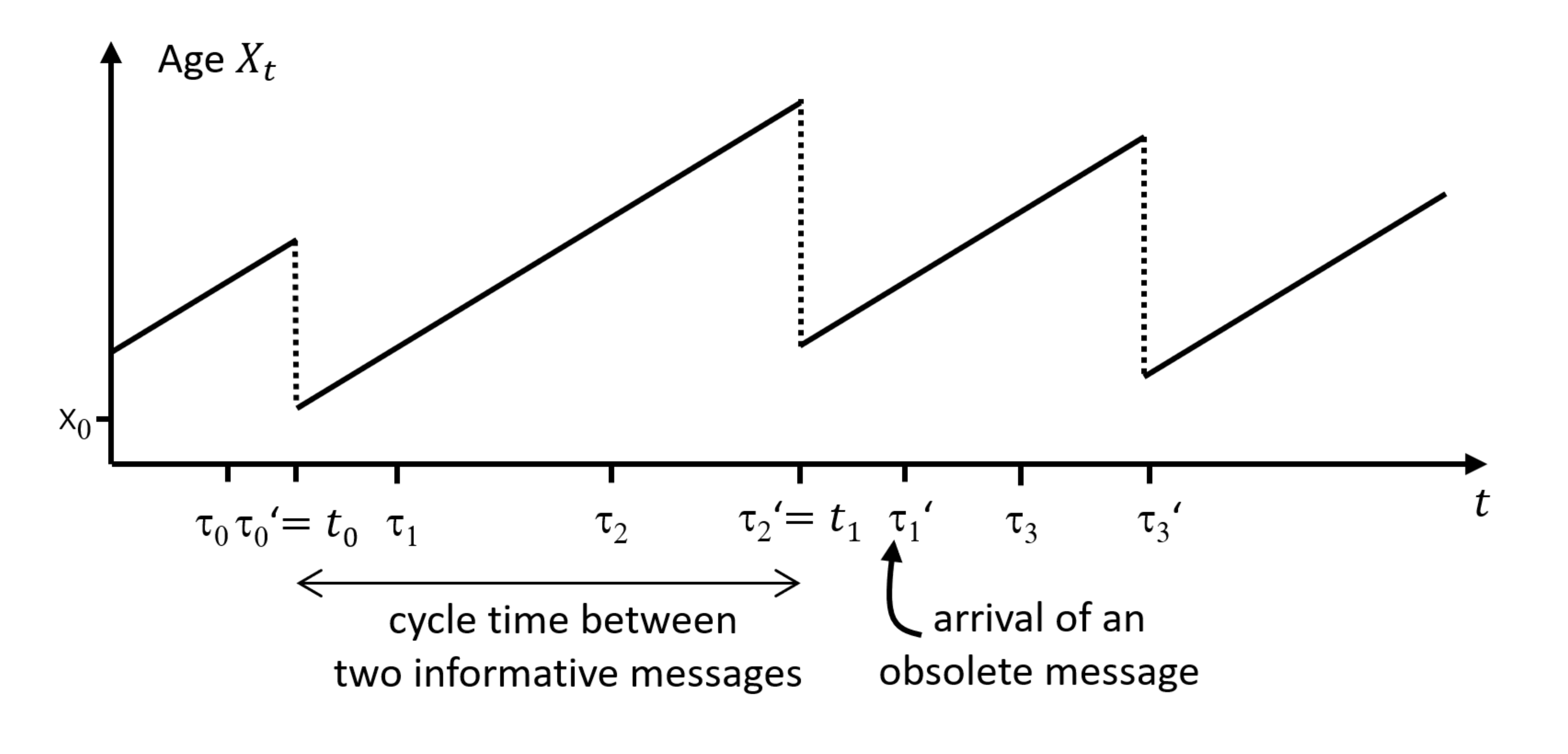}
  \caption{The age process $X_t$. Message $m_i$ is emitted at time $\tau_i$ and received at time $\tau'_i$. Observe that message $m_1$ is overtaken by message $m_2$ hence the age process at the receiver does not change when $m_1$ arrives.
  }
  \label{fig:age_sample_path}
  \vspace{-30pt}
\end{figure}

We assume that messages are generated according to a Poisson process of rate $\lambda$. The channel is modelled as a number of independent parallel servers each serving at most one message at a time with exponentially distributed service times, i.e. the random variables $\tau'_i-\tau_i$ are independent of each other and of the arrival process  $\{\tau_i\}$ and they are exponentially distributed with same parameter $\mu$.
Furthermore, in order to not overwhelm the channel the sender is window-flow-controlled and allows only a fixed number of outstanding informative messages $I_{\max}$ in the channel: arriving messages are dropped if the number of outstanding informative messages is equal to $I_{\max}$.  
We assume that the sender knows the number of informative messages in the channel (presumably via some instantaneous reverse channel).
%
We use the notation $M/M/I_{\max}/I_{\max}^*$ for this queueing system, where the $*$ here means that the departure of a message flushes all older messages out of the system.

In this paper we are interested in the stationary distribution of the age $X_t$ given the
process parameters $\lambda,\mu$ and $I_{\max}$.

The global notation used in the paper is recalled in Table~\ref{tab-nl}.
\begin{table}[tb]
  \caption{Notation List}
  \label{tab-nl}

  \centering
\begin{tabular}{|cp{6cm}|}
\hline
$\tilde{d}_n$, $d'_n$ & $\tilde{d}_n=\sum_{n'}Q_{n,n'}$, $d'_n=\sum_{n'}Q'_{n,n'}$\\
$f(x)$ & PDF of age of information at received at an arbitrary point in time;\\
$f^\circ(x_0,t_1)$ & Joint PDF of age of information $x_0$ and time to wait until next delivery of informative message, sampled when an informative message arrives at receiver;\\
  $f^\circ_A(x_0)$ & PDF of age of information sampled when an informative message arrives at receiver;\\
  $f_{n',n}$ & Laplace-Stieltjes Transform of $x_0\mapsto  g^\circ(x_0|n',n)$\\
  $\tilde{f}_n$ & Laplace-Stieltjes Transform of $t_1\mapsto  h(t_1|n)$\\
  $g^\circ(x_0|n',n)$ &PDF of the age $x_0$ just after an informative message arrival given that the state of the Markov chain is $n'$ just before the arrival of the informative message and $n$ just after the arrival;\\
  $h(t_1|n)$ & PDF of the time that will elapse from time $t$ until the next informative message arrives,  given that $Z_t=n$;\\
  $I_{\max}$ & Maximum number of messages in transit; messages generated when $Z_t=I_{\max}$ are discarded;\\
  $\lambda$ & Rate of generation of messages;\\
  $\mu$ & Message transit time is exponential with rate $\mu$;\\
  $\bar{N}$ &=$\sum_{i=1}^{I_{\max}}i p_{i}$\\
  $p_n$ & Stationary probability of $Z_t$\\
  $p^\circ_{n',n}$ & Probability that an arbitrary informative message arrival
happens at a transition $(n'\to n)$ of the Markov chain $Z_t$\\
$Q_{i,j}, Q'_{i,j}$ & Rate of transition of $Z_t$ [resp. $Z^r_t$] from state $i$ to state $j$\\
  $X_t$ & Age of information at receiver at time $t$;\\
  $Z_t$ & Number of messages in transit at time $t$;\\
  $Z^r_t$ & Time-reversed process derived from $Z_t$\\
  \hline
\end{tabular}
\vspace{-30pt}
\end{table}

\vspace{-20pt}
\section{A Palm Calculus approach to the AoI}
\label{sec:palm_model}

\subsection{The Underlying Queuing Model}
\label{sec:underlying_model}
First we consider a continuous time Markov jump process $\{Z_t\}_{t\ge0}$ 
that models the $M/M/I_{\max}/I_{\max}^*$ queue described in the previous section. Let $Z_t$ represent the number of messages underway from the sender to the receiver, for $t\in \mathbb{R}^+$, with $Z_t\in E=\{0,1, ..., I_{\max}\}$.
Recall that, by our modelling assumption, this counts
only informative messages. At any time $t$ such that $Z_t=n>0$, and for $i \in\{1, ..., n\}$ we call $i$th message, the message with the $i$th smallest timestamp among all messages present in the channel.

When the sender generates a new message at time $t$ (which occurs at constant rate $\lambda$), if $Z_{t^-}<I_{\max}$ then the message is accepted in the channel and $Z_t$ is incremented by $1$, i.e. $Z_{t^+}=Z_{t^-}+1$; else, i.e. if $Z_{t^-}=I_{\max}$, the message is discarded and $Z_t$ is unchanged.

Consider now message departures from the channel. Whenever $Z_t=n>0$ all $n$ messages in the channel can leave the channel with same rate $\mu$, thus the rate of message departure is $n\mu$ and all messages are equally likely to leave the channel. Assume that a departure occurs at time $t$ and $Z_{t^-}=n>0$. For $i\in \{1...n\}$, the probability that the departing message is the $i$th message is $\frac1n$. In this case, $Z_t$ is decremented by $i$, i.e. $Z_{t^+}=Z_{t^-}-i$; in other words, the transition $n\to n-i$ occurs at rate $\mu$ for every $i\in \{1...n\}$.

Thus $Z_t$ is a continuous-time Markov chain with finite state space $E$ and with transition rates (Fig.~\ref{fig:MC_fwd}):
\vspace{-10pt}
\begin{eqnarray}
Q_{i,i+1}&=\lambda, & i=0...I_{\max}-1\nonumber\\
Q_{i,j}&=\mu, & i=1...I_{\max}, 0\leq j \leq i-1\nonumber\\
Q_{i,j}&=0,& \mbox{otherwise.}
\label{eq:transition_rates_fwd_process}
\end{eqnarray}
\vspace{-10pt}

Observe that $Z_t$ can be regarded as the number of messages in a FIFO queue with Poisson arrivals of rate $\lambda$ and drained using a batch service process.  It is ergodic as the state space is finite and fully connected. 
%
%
%

Using the balance equation, the steady state probabilities $p_n$ can be computed and are given by:
\begin{equation}
p_n = \begin{cases}
\frac{(n+1) \lambda^n\mu}{\prod\limits_{j=1}^{n+1}\left(\lambda + j\mu\right)} &\text{for $0\leq n < I_{\max}$}\\ \\
\frac{\lambda^n}{\prod\limits_{j=1}^{n}\left(\lambda + j\mu\right)}&\text{for $n = I_{\max}$}
\end{cases}
\label{eq:steady_state_prob_chain_nr_msgs_underway}
\end{equation}
The derivation of \eqref{eq:steady_state_prob_chain_nr_msgs_underway} is given in appendix~\ref{sec:appendix_proof_steady_state_prob_fwd}.

\begin{figure}
  \centering
  \includegraphics[width=\linewidth]{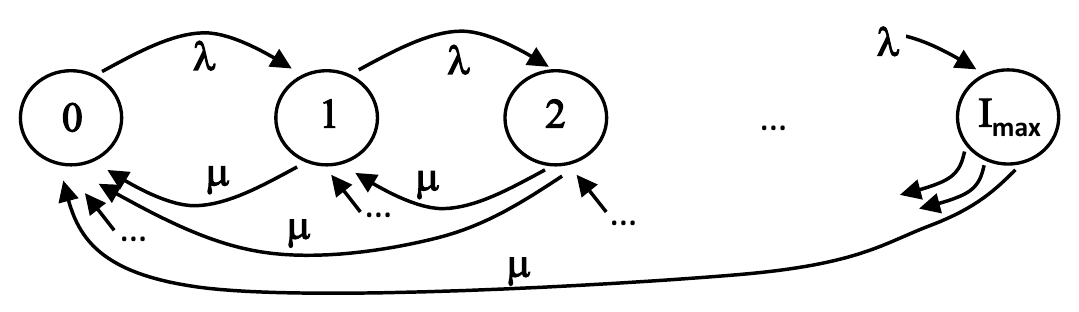}
  \caption{State transition diagram of the Markov chain $Z_t$ representing the number of non-obsolete messages underway.}
  \label{fig:MC_fwd}
\end{figure}
\subsection{A Palm Calculus Approach}
In the following we use Palm Calculus to compute the stationary distribution of age. To this end, we assume that the continuous time Markov chain $Z_t$ is in its unique stationary regime, which, since it is ergodic, occurs in practice if the system has been operating for a long time. With Palm calculus, we are able to relate the stationary distribution of the age to quantities that are computed for the Markov chain $Z_t$.

Palm calculus \cite{baccelli2012palm,serfozo2009basics,le2010performance} applies to a stationary point process $T_n$ ($n \in \mathbb{Z}$) and an observable (random)
process $X_t$ ($t\in\mathbb{R}$) that are jointly stationary. Here we take for $T_n$ the sequence of times at which a departure
occurs from the $M/M/I_{\max}/I_{\max}^*$ queue $Z_t$ (i.e. when $Z_t$ is decremented, which also corresponds to arrivals of informative messages at the receiver). Since we assume $Z_t$ is in its stationary
regime, this point process is also stationary. In the context of Palm calculus, it is customary to assume that
the numbering convention is such that $T_0\leq 0 < T_1$. The observable $X_t$ is the age of information at the receiver, as defined
earlier.
Note that $X_t$ can be computed in a deterministic way from the trajectory $Z_{(-\infty,t]}$ and is invariant with respect to change of time
origins, therefore it is jointly stationary with $Z_t$, hence with $T_n$ \cite[Section 7.2.1]{le2010performance}. Also note that Palm calculus does not require the point process to be Poisson (the arrival process is Poisson by definition, but it can be seen that the departure process is not).

\begin{figure}
  \centering
  \includegraphics[width=\linewidth]{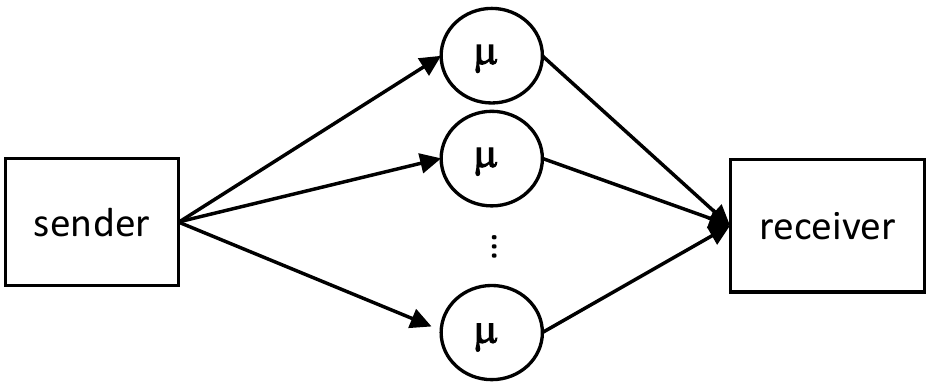}
  \caption{Our system model assumes that messages are transmitted upon generation and take random iid one way delay to reach the receiver. Hence the system model assumes for every message an independent channel each with exponentially distributed service time with identical parameter $\mu$. We assume a network channel (as sketched in Fig.~\ref{fig:cps_scenario}) that is constrained by a finite number of informative messages under way denoted by $I_{\max}$. This assumption corresponds to a window flow constrained sender with a maximum number of outstanding informative messages $I_{\max}$ given a perfect reverse channel.}
  \label{fig:system_model}
\end{figure}

We can now apply Palm's inversion formula \cite[Theorem 7.1]{le2010performance}, which states that, for any bounded, measurable test function $\varphi$ we have
\vspace{-5pt}
\begin{equation}
    \E\left[\varphi(X_t)\right] = \hat{\lambda} \E^\circ\left[\int_{T_0=0}^{T_1}\varphi(X_s) ds\right]
    \label{eq:palm_inversion_formula}
  \end{equation}
In the above,  $\E^\circ$ stands for the Palm expectation, which is the conditional expectation given that the point process has a point
at time $0$ (i.e. given that there is a departure from the $M/M/I_{\max}/I_{\max}^*$ queue at time $0$)\footnote{Note that the definition of the conditional expectation can be given a meaning even though the probability of the point process having a point at time $0$ exactly is $0$ \cite[Section 7.2.2]{le2010performance}.}. Also, under this conditional
expectation, $T_0=0$ and $T_1$ is the following departure instant. Last, $\hat{\lambda}$ is the intensity of the point process of departures,
which can be calculated from the Markov chain as $\hat{\lambda} = \sum_i i p_i  \mu = \mu \bar{N}$ with the stationary expectation of $Z_t$ denoted as $\bar{N}:=\sum_{i=1}^{I_{\max}}i p_{i}$.

 Observe that obtaining $\E\left[\varphi(X_t)\right]$ for arbitrary $\varphi$ is equivalent to
 finding the stationary distribution of the age of information at an arbitrary point in time.
Applying these ideas to the AoI gives the following theorem:

\begin{theorem}
The stationary PDF of the age of information at an arbitrary point in time, $f(x)$, is given by

\begin{equation}
f(x) = \hat{\lambda} \int_{x-x_0}^{\infty} \int_{0}^{x}  f^\circ(x_0,t_1) dx_0 dt_1
    \label{eq:f_age_fct_of_pdf_at_arrival}
\end{equation}
 where  $f^\circ(x_0, t_1)$ denotes the joint PDF of the age $x_0$ just after an informative message arrival \emph{and} of the time $t_1$ that will elapse  until the next informative message arrives.
\end{theorem}
Note that $f^\circ(x_0, t_1)$ is a Palm PDF, i.e. it corresponds to observations made upon the arrival of an informative message. Following the conventions in \cite{baccelli2012palm}, we use a $^\circ$ superscript to denote a Palm PDF.
\begin{proof}
We apply Palm's inversion formula \eqref{eq:palm_inversion_formula}.
Next, note that for $0\leq s \leq T_1$ we have $X_s=s+X_{0^+}$, therefore
\begin{align}
\E^\circ\left[\int_{0}^{T_1}\varphi(X_s) ds\right] = \E^\circ\left[\int_{0}^{T_1}\varphi(s+X_{0^+}) ds\right]
 \label{eq:plam_inversion_derivation1}
  \end{align}
 By definition of $f^\circ(x_0, t_1)$, it follows that
\begin{align}
    &\E^\circ\left[\int_{0}^{T_1}\varphi(X_s) ds\right] \nonumber \\
    &= \int_{0}^{\infty}  \int_{0}^{\infty} \int_{0}^{t_1} \varphi(x_0 + s) ds f^\circ(x_0,t_1) dx_0 dt_1\nonumber \\
    &= \int_{0}^{\infty}  \varphi(u)  \int_{0}^{u}  \int_{u-x_0}^{\infty} f^\circ(x_0,t_1) dt_1 dx_0  du \quad,
    \label{eq:plam_inversion_derivation}
\end{align}
  where, in the last line we substituted $u=x+s$ with $x\leq u \leq x+t$.


Now we find the PDF of the age at any arbitrary point in time $f(x)$ from comparing \eqref{eq:plam_inversion_derivation} with
    \begin{align}
 \E\left[\varphi(X_t)\right] = \int_{0}^{\infty} \varphi(x) f(x)  dx
\label{eq:expectation_test_fct_age}
  \end{align}
  where we find
\begin{equation}
f(x) = \hat{\lambda}  \int_{0}^{x}  \int_{x-x_0}^{\infty} f^\circ(x_0,t_1) dt_1 dx_0
    \label{eq:f_age_fct_of_pdf_at_arrival}
\end{equation}
%
%

\end{proof}
In the following, we calculate the Palm distribution $f^\circ(x_0,t_1)$. 
This is tractable because it involves $Z_t$ only, which is a Markov process on a finite state space. 
\section{Computing the Palm Distributions}
\label{sec:computing_palm_probabilities}

\subsection{Decomposition into Forward and Backward Components}
In order to compute the Palm PDF $f^\circ(x_0, t_1)$ we observe that the part on $x_0$ (the age) involves the past of $Z_t$ whereas the part on $t_1$ (time until a new arrival) involves the future of $Z_t$. This is captured by the following theorem.

\begin{theorem}
The joint PDF of both the age $x_0$ just after the informative message arrival \emph{and} the length of the cycle $t_1$ until the arrival of the next informative message is given by
\begin{align}
\label{eq:forward-backward2}
 f^\circ(x_0,t_1)= \sum_{(n',n) \mst 1\leq n+1\leq n'\leq I_{\max}} p^\circ_{n',n} g^\circ(x_0|n',n)h(t_1|n)
\end{align}
where
\begin{itemize}
  \item $p^\circ_{n',n}$ is the probability that an arbitrary message arrival
happens at a transition $(n'\to n)$ of the Markov chain $Z_t$ and is given by
\begin{equation}
  p^\circ_{n',n} = \frac{p_{n'}}{\bar{N}} \; \ind{1\leq n+1\leq n'\leq I_{\max}}
  \label{eq-p0}
\end{equation}
in the above, $\bar{N}$ is the stationary expectation of $Z_t$ and $p_i$ is given in \eqref{eq:steady_state_prob_chain_nr_msgs_underway};
  \item  $g^\circ(x_0|n',n)$ is the PDF of the Palm distribution of the age $x_0$ just after the informative message arrival given that the state of the Markov chain is $n'$ just before the arrival of the informative message and $n$ just after the arrival, where $n'\geq n+ 1$;
  \item  $h(t_1|n)$ is the stationary PDF of the time that will elapse from time $t$ until the next informative message arrives,  given that $Z_t=n$.

\end{itemize}
 \end{theorem}
The proof exploits the Markov property of $Z_t$, which expresses that the future depends on the past only through the present state.

 \begin{proof}
Define $f^\circ(x_0, t_1|n', n)$ as the joint PDF of the Palm distribution of
both the age $x_0$ just after the informative message arrival \emph{and}
the length of the cycle $t_1$ until the arrival of the next informative message, given that the state of the Markov chain is $n'$ just before the arrival of the informative message and $n$ just after the arrival, where $n'\geq n+ 1$.
It follows that the required PDF $f^\circ(x_0, t_1)$ is given by

\begin{align}
\label{eq:forward-backward}
 f^\circ(x_0,t_1)= \sum_{(n',n) \mst 1\leq n+1\leq n'\leq I_{\max}} p^\circ_{n',n} f^\circ(x_0, t_1|n',n)
\end{align}
where $p^\circ_{n',n}$  is the probability that an arbitrary message arrival
happens at a transition $(n'\to n)$ of the Markov chain $Z_t$. By \cite[Thm 7.1.2]{boudec2011performance}, such a probability is given by

\begin{equation}
  p^\circ_{n',n} = \eta p_{n'} Q_{n',n}, \; \ind{1\leq n+1\leq n'\leq I_{\max}}
\end{equation}
where $\mathbf{1}_{\{\cdot\}}$ is the indicator function, equal to $1$ when the condition is true and $0$ otherwise, $p_{n'}$ is the stationary probability given in \eqref{eq:steady_state_prob_chain_nr_msgs_underway}, $Q_{n',n}$ is the transition rate in \eqref{eq:transition_rates_fwd_process} and $\eta$ is a normalizing constant.  Observe that $Q_{n',n}=\mu$, which gives $\eta^{-1}=\sum_{i=1}^{I_{\max}}i p_{i}=\bar{N}$ where $\bar{N}$ is the stationary expectation of $Z_t$. It finally comes
\begin{equation}
  p^\circ_{n',n} = \frac{p_{n'}}{\bar{N}} \; \ind{1\leq n+1\leq n'\leq I_{\max}}
\end{equation}


Let $g^\circ(x_0|n',n)$ denote the PDF of the Palm distribution of the age $x_0$ just after the informative message arrival given that the state of the Markov chain is $n'$ just before the arrival of the informative message and $n$ just after the arrival, where $n'\geq n+ 1$.

Recall that
$h(t_1|n)$ denotes the stationary PDF of the time that will elapse from time $t$ until the next informative message arrives,  given that $Z_t=n$. By the Markov property, this is also the PDF of the Palm distribution of the time until the next informative message arrives given that the state of the Markov chain is $n'$ just before the arrival of the informative message and $n$ just after the arrival. Again by the Markov property,
$f(x_0, t_1|n'n)= g^\circ(x_0|n',n)h(t_1|n)$, which proves \eqref{eq:forward-backward2}.
\end{proof}

We next compute $h(t_1|n)$, which we call the forward component of \eqref{eq:forward-backward2}. The computation of the backward component $g^\circ(x_0|n',n)$ will involve a similar method plus a time-reversal argument.

\subsection{Computation of the Forward Component}

First, we will introduce the following lemma to calculate the Laplace-Stieltjes Transform (LST) of the time until the occurrence of the \emph{next transition of interest} in a continuous-time Markov chain $\{\tilde{Z}(t)\}_{t\in\mathbb{R_+}}$ conditioned on $Z_t=n$. The transitions of interest are defined by some subset $\tilde{\mathcal{F}}$ of $E\times E$, where $E\subseteq \mathbb{N}$ is the state-space of the Markov chain.

\begin{lemma}\label{lemma:LST_of_1st_event_forward_component_given_state_n}
Consider a time-homogeneous, continuous-time Markov chain $(\tilde{Z}_t)_{t\in\mathbb{R_+}}$ with state space $E\subseteq \Nats$
and with transition rates $\tilde{Q}_{n,n'}$;
let $\tilde{d}_{n}=\sum_{n'\in E}\tilde{Q}_{n,n'}$ denote the sum of all outgoing rates from state $n$ and
assume that $\tilde{d}_{n}>0$ for all $n\in E$. Let $\tilde{\mathcal{F}}\subseteq E$ such that
$\tilde{Q}_{n,n'}> 0$
for all $(n,n')\in \tilde{\mathcal{F}}$.

Call $\tilde{Y}_t$ the time that will elapse from $t$ until the next jump in $\tilde{\mathcal{F}}$ of the Markov chain, i.e.
$\tilde{Y}_t =\inf\{s> 0, (Z_{(t+s)^-},Z_{(t+s)^+})\in \tilde{\mathcal{F}}\}$.
The conditional LST of $\tilde{Y}_t$ given that $\tilde{Z}_t=n$, denoted as $\tilde{f}_{n}(\nu)$, satisfies
\begin{align}
&\tilde{f}_{n}(\nu) \coloneqq \E\left[e^{-\nu \tilde{Y}_t} | \tilde{Z}_t=n\right] \nonumber\\
&= \frac{1}{\tilde{d}_{n}+ \nu}\left(\sum_{\substack{n',\\(n,n')\notin\tilde{\mathcal{F}}}} \tilde{f}_{n'}(\nu) \tilde{Q}_{n,n'}+ \sum_{\substack{n',\\(n,n')\in\tilde{\mathcal{F}}}} \tilde{Q}_{n,n'}\right).
\label{eq:LST_of_next_transition_of_interest}
\end{align}
\end{lemma}
\begin{proof}
%
Fix some arbitrary time $t$ and define $\tilde{S}_t$ as the time until the next transition (of interest or not) out
of state $\tilde{Z}_t$ and let $ N'_t\coloneqq Z_{t+\tilde{S}_t}$ denote the next state. It is known \cite{gillespie1976general} that, conditional to $\tilde{Z}_t=n$, $N'_t$ and $\tilde{S}_t$ are independent, the distribution of $\tilde{S}_t$ is exponential with rate
$\tilde{d}_{n}$ and the distribution of $N'_t$ is given by $\P(N'_t=n'| \tilde{Z}_t=n)=\frac{\tilde{Q}_{n,n'}}{\tilde{d}_n}$. It follows that
\begin{equation}
\P\left[N'_t=n'| \tilde{Z}_t=n,\tilde{S}_t=s\right]=\frac{\tilde{Q}_{n,n'}}{\tilde{d}_n}
\label{eq:gillespie}
\end{equation}
and
\begin{equation}
\E\left[e^{-\nu \tilde{S}_t}\right] = \frac{\tilde{d}_{n}}{\tilde{d}_{n}+ \nu}
\label{eq:gillespie2}
\end{equation}

Also let $\tilde{R}_t$ denote the residual time from the next transition until the next transition of interest, i.e. $\tilde{R}_t=0$ whenever $(\tilde{Z}_t,N'_t)\in \tilde{\mathcal{F}}$ and
otherwise 
$\tilde{R}_t=\tilde{Y}_{t+\tilde{S}_t}$. Hence
%
%
\begin{equation}
\tilde{Y}_t = \tilde{S}_t + \tilde{R}_t
\label{eq:Y_first_departure}
\end{equation}
%

By conditioning on $\tilde{S}_t=s$ we can write
\begin{align}
&\E\left[e^{-\nu \tilde{Y}_t} | \tilde{Z}_t=n,\tilde{S}_t=s\right] \nonumber\\
&= e^{-\nu s}\E\left[e^{-\nu \tilde{R}_t} | \tilde{Z}_t=n,\tilde{S}_t=s\right]
\label{eq:Y_lst_1}
\end{align}

By conditioning with respect to $N'_t$ in the latter term and applying \eqref{eq:gillespie} we obtain
\begin{align}
&\E\left[e^{-\nu \tilde{R}_t} | \tilde{Z}_t=n,\tilde{S}_t=s\right] \nonumber\\
&= \sum_{n'\in E}\left(\E\left[e^{-\nu \tilde{R}_t} | \tilde{Z}_t=n,\tilde{S}_t=s, N'_t=n'\right]\times \right.\nonumber\\
&\left.\P\left[N'_t=n'| \tilde{Z}_t=n,\tilde{S}_t=s\right]\right)\nonumber\\
&=\sum_{n'\in E}\left(\E\left[e^{-\nu \tilde{R}_t} | \tilde{Z}_t=n,\tilde{S}_t=s, N'_t=n'\right]\frac{\tilde{Q}_{n,n'}}{\tilde{d}_n}\right)
\label{eq:Y_lst_1b}
\end{align}

Now if $(n,n')\in \tilde{\mathcal{F}}$ then $\tilde{R}_t=0$ hence
\begin{equation}
\E\left[e^{-\nu \tilde{R}_t} |\tilde{Z}_t=n,\tilde{S}_t=s, N'_t=n'\right]=1  \mbox{ if } (n,n')\in \tilde{\mathcal{F}}
\label{eq:Y_lst_1c}
\end{equation}

Else, i.e. if $(n,n')$ is not in $\tilde{\mathcal{F}}$, $\tilde{R}_t=\tilde{Y}_{t+\tilde{S}_t}$ is the time that remains
to elapse until the next transition of interest; by the Markov property,
the future of the Markov chain depends on the history only via the current state, i.e.
\begin{align}
&\E\left[e^{-\nu \tilde{Y}_{t+\tilde{S}_t}} | \tilde{Z}_t=n,\tilde{S}_t=s, N'_t=n'\right] \nonumber\\
&
=\E\left[e^{-\nu \tilde{Y}_{t+s}} | N'_t=n'\right]=  \E\left[e^{-\nu \tilde{Y}_{t+s}} | Z_{t+s}=n'\right] \nonumber\\
&=\tilde{f}_{n'}(\nu)
\label{eq:Y_lst_1d}
\end{align}
where the last equality is because the Markov chain is time-homogeneous.

Combining \eqref{eq:Y_lst_1} with \eqref{eq:Y_lst_1b}-\eqref{eq:Y_lst_1d} gives
\begin{align}
&\E\left[e^{-\nu \tilde{Y}_t} | \tilde{Z}_t=n,\tilde{S}_t=s\right]\nonumber\\
&= e^{-\nu s}\left(\sum_{\substack{n',\\(n,n')\notin\tilde{\mathcal{F}}}} \tilde{f}_{n'}(\nu) \frac{\tilde{Q}_{n,n'}}{\tilde{d}_{n}}+ \sum_{\substack{n',\\(n,n')\in\tilde{\mathcal{F}}}} \frac{\tilde{Q}_{n,n'}}{\tilde{d}_{n}}\right)
\label{eq:Y_lst_1f}
\end{align}
By the law of total expectation we can now write
\begin{align}
&\tilde{f}_{n}(\nu) = \E\left[e^{-\nu \tilde{Y}_t} | \tilde{Z}_t=n\right] \nonumber\\
&=\E\left[\E\left[e^{-\nu \tilde{Y}_t} | \tilde{Z}_t=n,\tilde{S}_t\right]\right] \nonumber\\
&= \E\left[e^{-\nu \tilde{S}_t}\right]\left(\sum_{\substack{n',\\(n,n')\notin\tilde{\mathcal{F}}}} \tilde{f}_{n'}(\nu) \frac{\tilde{Q}_{n,n'}}{\tilde{d}_{n}}+ \sum_{\substack{n',\\(n,n')\in\tilde{\mathcal{F}}}} \frac{\tilde{Q}_{n,n'}}{\tilde{d}_{n}}\right)
\label{eq:LST_of_next_transition_of_interest2}
\end{align}
Using \eqref{eq:gillespie2} completes the proof.
\end{proof}
\vspace{5pt}

Now we can use Lem.~\ref{lemma:LST_of_1st_event_forward_component_given_state_n} to calculate the stationary PDF $h(t_1|n)$ of the time that will elapse from a fixed time $t$ until the next informative message arrives conditioned on $Z_t=n$.
The set of transitions of interest is $\tilde{\mathcal{F}}\coloneqq \{(i,j)\}_{i>j}$, i.e. the transitions associated with the arrival of informative messages. 
The transition rates $Q_{i,j}$ are given in \eqref{eq:transition_rates_fwd_process} and

\begin{equation}
\label{eq-dtilde}\tilde{d}_{n}= \lambda \mathbf{1}_{\{n<I_{\max}\}} + \mu n\end{equation}

 where $\mathbf{1}_{\{\cdot\}}$ is the indicator function, equal to $1$ when the condition is true and $0$ otherwise.

%

The Laplace-Stieltjes Transform of $h(t_1|n)$ continues to be denoted by $\tilde{f}_{n}(\nu)$; the application of Lem.~\ref{lemma:LST_of_1st_event_forward_component_given_state_n} gives:
\begin{align}
\tilde{f}_{n}(\nu) = \frac{1}{\tilde{d}_{n}+ \nu} \left[  n \mu +  \lambda\tilde{f}_{n+1}(\nu)\right].
\label{eq:LST_of_next_arrival_given_state_nprime}
\end{align}
\begin{figure}[t]
  \centering
  \includegraphics[width=0.9\linewidth]{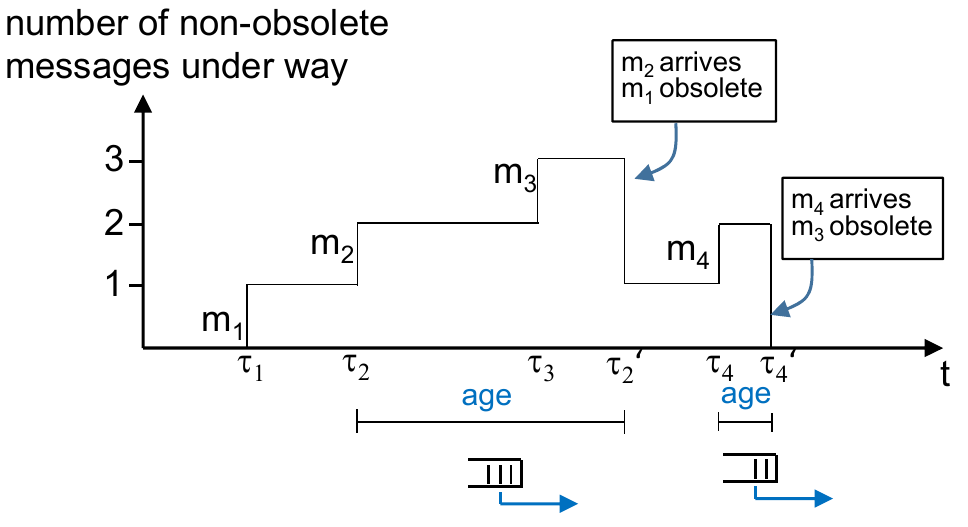}
  \caption{A sample path of the number of non-obsolete messages under way given the system model from Sect.~\ref{sec:system_model}. Looking at the forward process: The downward jumps mark the arrival of informative messages at the receiver which make previous messages obsolete. As messages depart in batches of random sizes in this model the waiting time of the freshest message of the batch corresponds to the age value set upon the arrival of that message at the receiver.}
  \label{fig:sample_path_jumps}
\end{figure}
%
This recursive relation can be rewritten using matrix notation as
\begin{equation}
\mt{(\nu I + \tilde{D}) \tilde{f} = \pmb{\bar{\mu}} + \bar{\Lambda} \tilde{f}}
\label{eq:LST_of_next_arrival_given_state_nprime_matrix_recursive}
\end{equation}
with the identity matrix $\mt{I}$, the vectors  $\mt{\tilde{f}}= [\tilde{f}_{0}(\nu),\dots,\tilde{f}_{I_{\max}}(\nu)]\T$, and $\pmb{\bar{\mu}} = [0,\mu,2\mu\dots,I_{\max}\mu]\T$, and the matrices $\mt{\tilde{D}} \coloneqq \diag (\tilde{d}_0,\tilde{d}_1,\dots,\tilde{d}_{I_{\max}})$,
and
\begin{align}
\pmb{\bar{\Lambda}}
\coloneqq
\begin{bmatrix}
  \mathbf{0} &   \lambda& \hdots& \mathbf{0}  \\
\vdots &  \ddots& \lambda& \vdots \\
\vdots& \ddots& \ddots &\lambda\\
\mathbf{0} & \hdots & \mathbf{0}&\mathbf{0}
\end{bmatrix}.
\label{eq:definition_of_Lambda_bar}
\vspace{-20pt}
\end{align}
Now we can directly solve for the conditional LSTs as
\begin{equation}
\mt{\tilde{f} = (\nu I + \tilde{D} - \bar{\Lambda})^{-1} \pmb{\bar{\mu}} }
\label{eq:LST_of_next_arrival_given_state_nprime_matrix}
\end{equation}

\subsection{Computation of the Backward Component}

Recall that $g^\circ(x_0|n',n)$ denotes the PDF of the Palm distribution of the age $x_0$ just after the informative message arrival
given that the state of the Markov chain is $n'$ just before the arrival of the informative message and $n$ just after the arrival, where $n'\geq n+ 1$.
since the arrival of the freshest message in the served batch.
For the computation of $g^\circ(x_0|n',n)$ we resort to time-reversal as this allows to use a similar method as for the forward component.

The time-reversed process $Z_t^r$ is defined by $Z_t^r=Z_{-t}$.
In a nutshell, time reversal allows us to change the underlying queueing model from Sect.~\ref{sec:underlying_model} into a FIFO queue where arrivals occur in message batches of random size while the server removes exactly \emph{one} message on each visit.
To illustrate this, consider the sample path shown in Fig.~\ref{fig:sample_path_jumps} in the reverse time direction.

It is shown in \cite{Kelly:Reversibility-2011} [Section 1.7] that if $Z_t$ is endowed with its stationary probability, then the
time-reversed process is also a time-homogeneous continuous time Markov chain with same state space and same stationary probability,
but with different transitions rates. Specifically, by \cite{Kelly:Reversibility-2011} [Theorem 1.12] the transition rates $Q'_{i,j}$
for $Z_t^r$ depend on the transition rates of the original Markov process \eqref{eq:transition_rates_fwd_process} and its stationary
distribution \eqref{eq:steady_state_prob_chain_nr_msgs_underway}. We obtain $Q'_{i,i-1} =\lambda_{i}'$ for $i=1...I_{\max}$, $Q'_{i,j} =\mu_{ij}'$ for $i=0...I_{\max}-1,i<j\leq I_{\max}$ and $Q'_{i,j} =0$ otherwise, with
\begin{equation}
\lambda_i' = \begin{cases}
\frac{i}{i+1}\left(\lambda + (i+1) \mu\right) &\text{for $1\leq i < I_{\max}$}\\ \\
I_{\max} \mu &\text{for $i = I_{\max}$}
\end{cases}
\label{eq:lambda_prime}
\end{equation}
and
\begin{equation}
\mu_{ij}' = \begin{cases}
 \frac{(j+1) \mu \lambda^{j-i}}{(i+1)\prod\limits_{k=i+2}^{j+1}\left(\lambda + k\mu\right)}&\text{for $j \neq I_{\max}$}\\ \\
\frac{\lambda^{I_{\max}-i}}{(i+1)\prod\limits_{k=i+2}^{I_{\max}}\left(\lambda + k\mu\right)}&\text{for $j = I_{\max}$}
\end{cases}
\label{eq:mu_prime}
\end{equation}
for $0\leq i< I_{\max}$. The derivation of \eqref{eq:lambda_prime}, \eqref{eq:mu_prime} is given in the appendix.


In $Z_t$, upon serving a batch of messages, the sojourn time of the freshest message of that batch is the age of that particular informative message. In $Z^r_t$, this is given by the sojourn time of the $(n+1)$st message of an arriving batch of size $n'-n$. In $Z^r_t$, the arrival of a batch corresponds to a transition $n\to n'$ with $n'\geq n+1$ and the size of the arriving batch is $n'-n$. It follows that, for $n'\geq n+ 1$, $g^\circ(x_0|n',n)$ can be re-interpreted as the PDF of the time from now until the $(n+1)$st departure of $Z^r_t$, given that $Z^r_t$ is doing a transition $n\to n'$ now.
Since $Z^r_t$ is also Markov, we can apply the Markov property and obtain that this is
simply the PDF of the time from now until the $(n+1)$st departure of $Z^r_t$
given that $Z^r_t=n'$. For $n+1=1$ this is the conditional PDF of the time until a next departure, which is exactly the problem that is solved in Lemma~\ref{lemma:LST_of_1st_event_forward_component_given_state_n}, and which we now extend as follows (the proof is similar and is not given).

\begin{lemma}\label{lemma:LST_of_all_events_forward_component_given_state_n}
Consider a time-homogeneous, continuous-time Markov chain $(\tilde{Z}_t)_{t\in\mathbb{R_+}}$ with state space $E\subseteq \Nats$
and with transition rates $\tilde{Q}_{n,n'}$;
let $\tilde{d}_{n}=\sum_{n'\in E}\tilde{Q}_{n,n'}$ denote the sum of all outgoing rates from state $n$ and
assume that $\tilde{d}_{n}>0$ for all $n\in E$. Let $\tilde{\mathcal{F}}\subseteq E$ such that
$\tilde{Q}_{n,n'}> 0$
for all $(n,n')\in \tilde{\mathcal{F}}$.

For $k\geq 1$, call $\tilde{Y}^k_t$ the time that will elapse from $t$ until the $k$th jump in $\tilde{\mathcal{F}}$ of the Markov chain, i.e.
$\tilde{Y}^1_t =\inf\{s\geq 0, (\tilde{Z}_{(t+s)^-},\tilde{Z}_{(t+s)^+})\in \tilde{\mathcal{F}}\}$ and for $k\geq 2$, $\tilde{Y}^k_t =\inf\{s> Y^{k-1}, (\tilde{Z}_{(t+s)^-},\tilde{Z}_{(t+s)^+})\in \tilde{\mathcal{F}}\}$.
The conditional LST of $\tilde{Y}^k_t$ given that $\tilde{Z}_t=n$, denoted as $\tilde{f}_{n,k}(\theta)$, satisfies, for $k\geq 1$:
\begin{align}
&\tilde{f}_{n,k}(\theta) = \frac{1}{\tilde{d}_{n}+ \theta} \times \nonumber\\
&\left(\sum_{\substack{n',\\(n,n')\notin\tilde{\mathcal{F}}}} \tilde{f}_{n',k}(\theta) \tilde{Q}_{n,n'}+ \sum_{\substack{n',\\(n,n')\in\tilde{\mathcal{F}}}} \tilde{f}_{n',k-1}(\theta)\tilde{Q}_{n,n'}\right).
\label{eq:LST_of_next_transition_of_interest}
\end{align}
where $\tilde{f}_{n,0}(\theta)=1$ by convention.
\end{lemma}
%



Let $f_{n',k}(\theta)$ be the LST of the time from now until the $k$th departure of $Z^r_t$
given that $Z^r_t=n'$, so that the LST of $g^\circ(x_0|n',n)$ is
$f_{n', n+1}$. To compute $f_{n',k}(\theta)$, we now apply
Lemma~\ref{lemma:LST_of_all_events_forward_component_given_state_n}
to the Markov chain $Z^r_t$, with transition rate matrix $Q'$, and obtain:

\begin{eqnarray}
\hspace{-20pt}f_{n',1}(\theta) \left[d'_{n'} +\theta\right]\mkern-18mu&=&\mkern-18mu \lambda'_{n'}\ind{n'>0} + \mkern-12mu\sum\limits_{j>n'} \mu_{n',j}' f_{j,1}(\theta)\label{eq:LST_of_next_departure_given_state_n_v2a}\\
\hspace{-20pt} f_{n',k}(\theta) \left[d'_{n'} +\theta\right]\mkern-18mu&=&\mkern-18mu \lambda'_{n'}\ind{n'>0} f_{n',k-1}(\theta)+ \mkern-11mu\sum\limits_{j>n'} \mu_{n',j}' f_{j,1}(\theta)
\label{eq:LST_of_next_departure_given_state_n_v2b}
\end{eqnarray}
for $0\leq n'\leq I_{\max}$ and $k\geq 2$. In the above, $\lambda'$ and $\nu'$ are given in
\eqref{eq:lambda_prime} and \eqref{eq:mu_prime} and
\begin{equation}
d'_{n'} = \sum_{n=0}^{I_{\max}}Q'_{n',n}
\label{eq-dprime}
\end{equation}

We use the following matrix notation: $\mt{D} \coloneqq \diag (d'_0,d'_1,\dots,d'_{I_{\max}})$,
$\mt{f}_{\cdot,k} = [f_{0,k}(\theta),\dots,f_{I_{\max},k}(\theta)]\T$,
 $\pmb{\lbar'} = [0,\lambda_1',\dots,\lambda_{I_{\max}}']\T$. $\mt{M}$ is the upper triangular matrix
\begin{equation*}
\mt{M} = \begin{cases*}
\mu_{ij}'&\text{for $i < j$}\\
0&\text{for $i \geq j$}
\end{cases*}
\end{equation*}
and
$\pmb{\Lambda}$ is the matrix with $\lambda_n'$ on the subdiagonal defined by
\begin{align}
\pmb{\Lambda}
\coloneqq
\begin{bmatrix}
  \mathbf{0} &  \hdots& \hdots& \mathbf{0}  \\
\lambda_1' & \ddots & & \vdots \\
\vdots& \ddots& \ddots &\vdots\\
\mathbf{0} & \hdots & \lambda_{I_{\max}}'&\mathbf{0}
\end{bmatrix}.
\label{eq:definition_of_Lambda}
\end{align}

for $i,j \in \{0,1,\dots,I_{\max}\}$,
We can rewrite the recursive relation of the conditional LST in  \eqref{eq:LST_of_next_departure_given_state_n_v2a} as
\begin{equation}
\mt{(\theta I + D) f_{\cdot,1} = \pmb{\lbar'} + M f_{\cdot,1}}
\label{eq:LST_of_next_departure_given_state_n_matrix_recursive}
\end{equation}
The previous equation can be solved and we obtain:
\begin{equation}
\mt{f_{\cdot,1} = (\theta I + D - M)^{-1} \pmb{\lbar'} }
\label{eq:LST_of_next_departure_given_state_n_matrix}
\end{equation}
Similarly, we can re-write \eqref{eq:LST_of_next_departure_given_state_n_v2b} as
\begin{equation}
\mt{(\theta I + D) f_{\cdot,k} = \pmb{\Lambda}f_{\cdot,k-1} + M f_{\cdot,k}}
\label{eq:LST_of_only_kth_departure_given_state_n_matrix_recursive}
\end{equation}
for $k\geq 2$. 
Now we can construct a block matrix form that takes \eqref{eq:LST_of_next_departure_given_state_n_matrix_recursive} as well as \eqref{eq:LST_of_only_kth_departure_given_state_n_matrix_recursive} to follow the form
\begin{align}
\scalemath{0.83}{
\begin{bmatrix}
\mt{\theta I + D} &  & \mathbf{0}  \\
& \ddots & \\
\mathbf{0} & &  \mt{\theta I + D}
\end{bmatrix}
\begin{bmatrix}
\mt{f_{\cdot,1}} \\
\vdots \\
\mt{f_{\cdot,I_{\max}}}
\end{bmatrix}
=
\begin{bmatrix}
\pmb{\lbar'}  \\
\vdots \\
\mathbf{0}
\end{bmatrix}
+
\begin{bmatrix}
\mt{M} &  & \mathbf{0}  \\
\pmb{\Lambda} & \ddots & \\
\mathbf{0} &\pmb{\Lambda}  & \mt{M}
\end{bmatrix}
\begin{bmatrix}
\mt{f_{\cdot,1}} \\
\vdots \\
\mt{f_{\cdot,I_{\max}}}
\end{bmatrix}
}
\label{eq:LST_of_kth_departure_given_state_n_matrix_recursive}
\end{align}

We can directly find the vector of conditional LST as
\begin{align}
\begin{bmatrix}
\mt{f_{\cdot,1}} \\
\vdots \\
\mt{f_{\cdot,I_{\max}}}
\end{bmatrix}
=
\begin{bmatrix}
\mt{\theta I + D-M} &  & \mathbf{0}  \\
\pmb{-\Lambda} & \ddots & \\
\mathbf{0} &\pmb{-\Lambda}  & \mt{\theta I + D-M}
\end{bmatrix} ^{-1}
\begin{bmatrix}
\pmb{\lbar'}  \\
\vdots \\
\mathbf{0}
\end{bmatrix}
\label{eq:LST_of_kth_departure_given_state_n_matrix}
\end{align}

Since the computation of \eqref{eq:LST_of_kth_departure_given_state_n_matrix} requires the inversion of a matrix of the order of $I_{\max}^2\times I_{\max}^2$ we show in the following how to calculate the conditional LST recursively from \eqref{eq:LST_of_kth_departure_given_state_n_matrix_recursive}.
We observe that $\mt{M- D  = Q' - \Lambda}$ where  $\mt{Q'}$ denotes the transition rate matrix of the continuous Markov chain associated with the reversed process $Z_t^r$.
We define $\mt{\Phi \coloneqq \theta I + D - M}$ and obtain the following recursion in block matrix form
\begin{align*}
\begin{bmatrix}
\mt{\Phi} &  & &  \\
\pmb{-\Lambda} & \ddots & \mbox{\Large $\mathbf{0}$}&\\
\mathbf{0} & \ddots& \ddots & \\
\mathbf{0} & \mathbf{0} & \pmb{-\Lambda}  & \mt{\Phi}
\end{bmatrix}
\begin{bmatrix}
\mt{f_{\cdot,1}} \\
\vdots \\
\vdots \\
\mt{f_{\cdot,I_{\max}}}
\end{bmatrix}
=
\begin{bmatrix}
\pmb{\lbar'}  \\
\vdots \\
\vdots \\
\mathbf{0}
\end{bmatrix}
\end{align*}
Now we obtain the conditional LST $\mt{f_{\cdot,n}}$ recursively with the initial condition
\begin{align}
\mt{f_{\cdot,1} = \Phi^{-1} \pmb{\lbar'}}
\label{eq:LST_of_1st_departure_given_state_n_initial_condition}
\end{align}
and for $k\geq 2$
\begin{align}
\mt{f_{\cdot,k} =  \Psi^{k-1}\Phi^{-1} \pmb{\lbar'}}
\label{eq:LST_of_kth_departure_given_state_n_rest_of_recursion}
\end{align}
where we used the shorthand notation $\mt{\Psi \coloneqq \Phi^{-1}\Lambda}$.

\begin{figure*}
\centering
  \begin{subfigure}[b]{0.49\linewidth}
    \includegraphics[width=\linewidth]{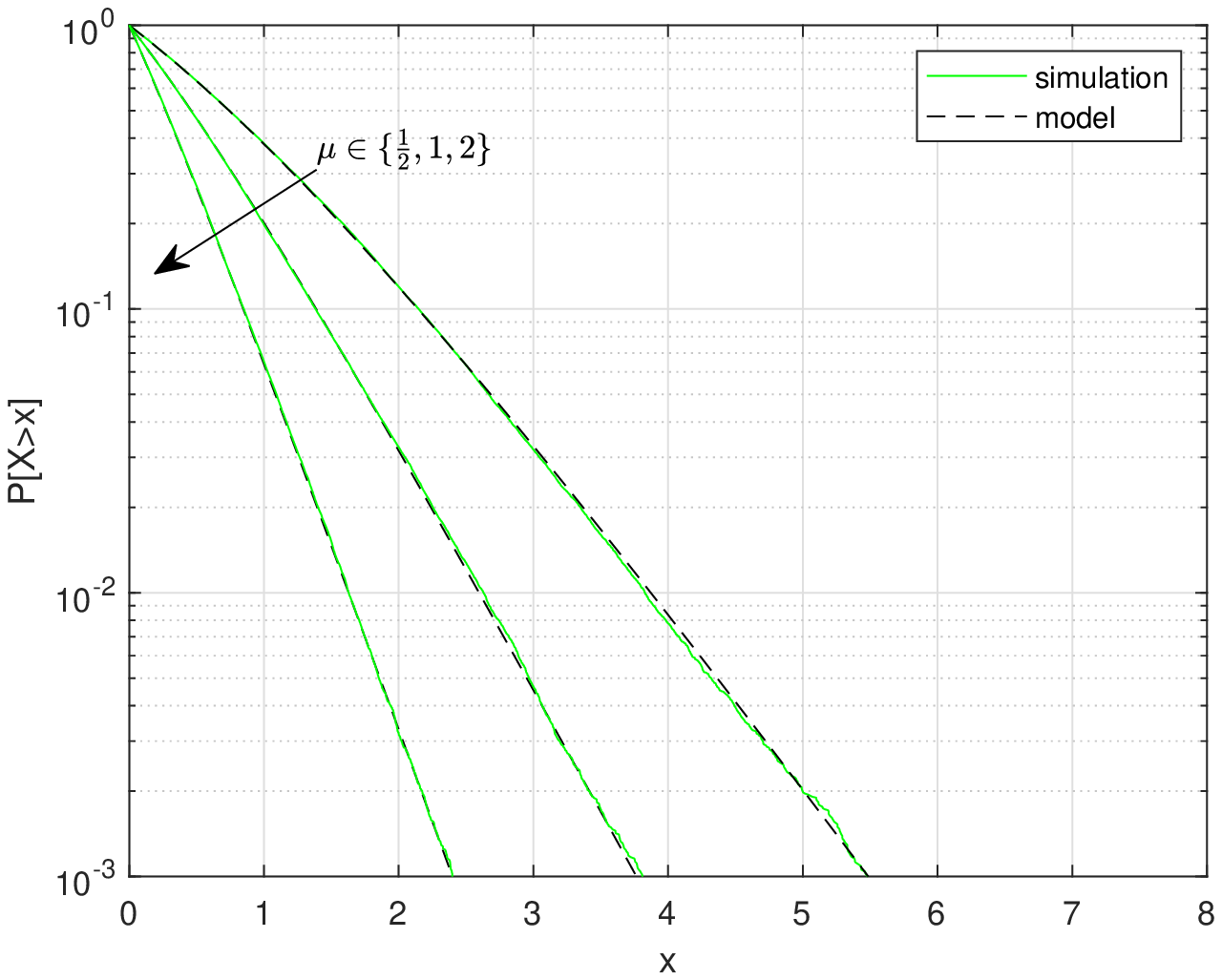}
    \caption{}
    \label{fig:ccdf_age_at_inform_2}
  \end{subfigure}
  \hfill
  \begin{subfigure}[b]{0.49\linewidth}
    \includegraphics[width=\linewidth]{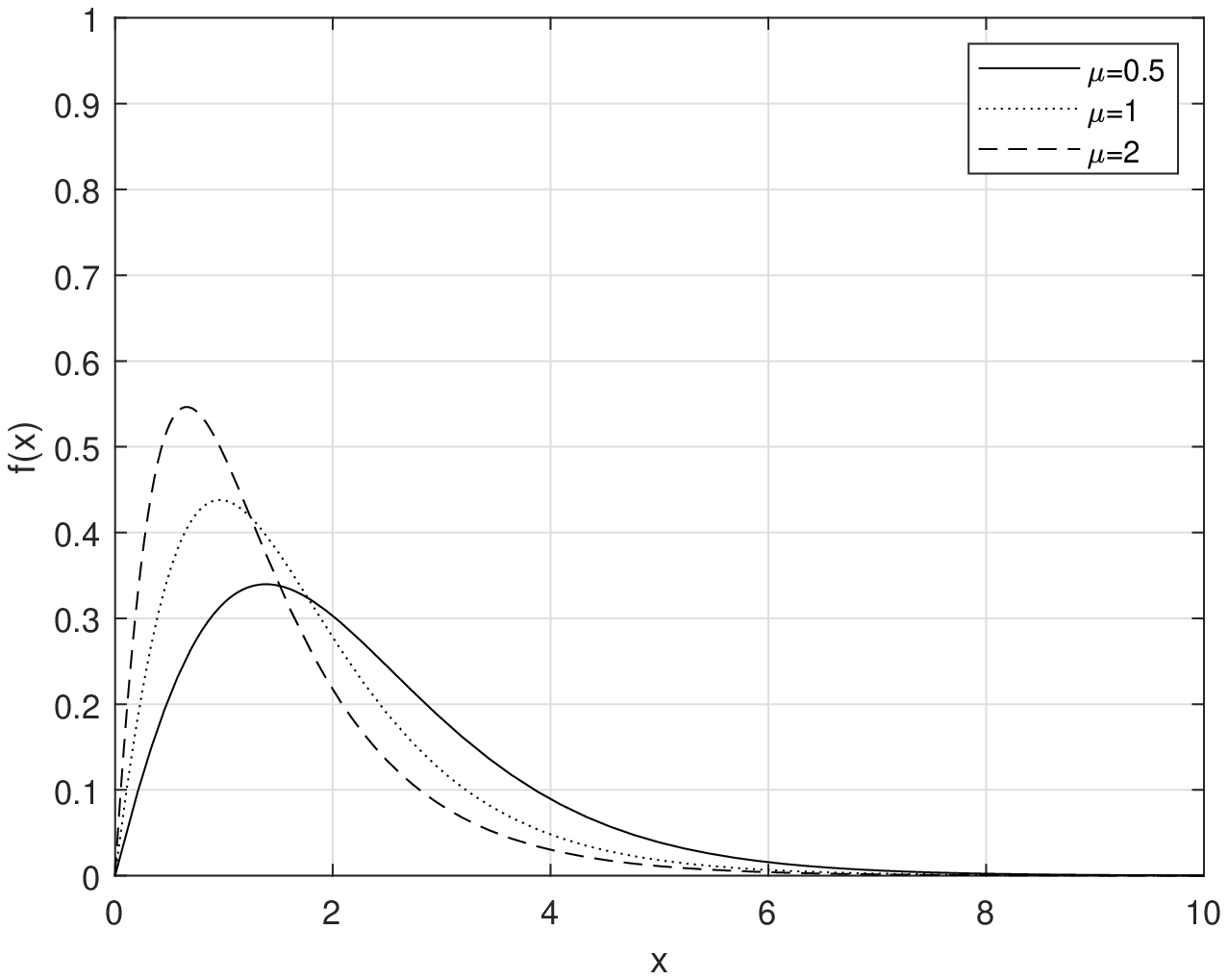}
    \caption{}
    \label{fig:pdf_age_at_anytime}
  \end{subfigure}
 \caption{(a) CCDF of the age at the arrival times of informative messages for arrival rate $\lambda=1$ and varying OWD parameter $\mu$. $I_{\max} = 20$. (b) Probability density of the age $f(x)$ at any point in time obtained from \eqref{eq:f_age_fct_of_pdf_at_arrival} for $\lambda=1$ and varying OWD parameter $\mu$. $I_{\max} = 20$.}
\end{figure*}

\subsection{Computing $f^\circ(x_0, t_1)$}
We can now put together the forward and backward elements. Let
$\hat{f}(\nu,\theta)$ denote the LST of $f^\circ(x_0,t_1)$, specifically
\begin{equation*}
  \hat{f}(\nu,\theta):=\int_{0}^{\infty}\int_{0}^{\infty}
  f^\circ(x_0,t_1) e^{-\nu t_1} e^{-\theta x_0} dt_1 dx_0
\end{equation*}
From \eqref{eq:forward-backward2} this becomes
\begin{align}
\label{eq:joint_LST_x_t}
\hat{f}(\nu,\theta)=\sum_{(n',n) \mst 1\leq n+1\leq n'\leq I_{\max}} p^\circ_{n',n}
f_{n',n+1}(\theta)\tilde{f}_n(\nu)
\end{align}
where $p^\circ_{n',n}$ is in \eqref{eq-p0},  $\tilde{f}_n(\nu)$ is the $n$th component of \eqref{eq:LST_of_next_arrival_given_state_nprime_matrix} and $f_{n',n+1}(\theta)$ is obtained by setting $k=n+1$ in \eqref{eq:LST_of_kth_departure_given_state_n_rest_of_recursion}.

As all the Laplace-Stieltjes transforms encountered here are rational fractions in
$\theta$ [resp. $\nu$], the distributions associated with them are matrix-exponential and can be computed in closed form given $\lambda, \mu$.
Specifically, we obtain
\begin{eqnarray}
  g^\circ(x_0|n',n)= \sum_{i=0}^{I_{\max}}\pi^{i,n',n}(x_0)e^{-x_0d'_i}
  \label{eq-g0-fin}
\end{eqnarray}
where $d'_i$ is given in \eqref{eq-dprime} and $\pi^{i,n',n}$ is a polynomial in $x_0$, the coefficients of which are computed numerically for every $(\lambda, \mu)$.
Similarly, we obtain
\begin{eqnarray}
  h(t_1|n)= \sum_{j=0}^{I_{\max}}\tilde{\pi}^{j,n}(t_1)e^{-t_1\tilde{d}_j}
  \label{eq:ht1_given_n}
\end{eqnarray}
where $\tilde{d}_j$ is given in \eqref{eq-dtilde} and $\tilde{\pi}^{j,n}$ is a polynomial in $t_1$, the coefficients of which are computed numerically for every $(\lambda, \mu)$.
Putting things together we obtain
\begin{align}
  f^\circ(x_0, t_1)&=\sum_{i=0}^{I_{\max}}\sum_{j=0}^{I_{\max}}e^{-x_0d'_i-t_1\tilde{d}_j}\label{eq-f0-fin}\\
  &\sum_{(n',n) \mst 1\leq n+1\leq n'\leq I_{\max}} p^\circ_{n',n} \pi^{i,n',n}(x_0) \tilde{\pi}^{j,n}(t_1)
  \nonumber
\end{align}




\section{Computing Age Performance Metrics}
\subsection{Age Distribution at Arbitrary Points in Time}

To obtain the PDF of the age at any point in time we insert the formulation of the PDF \eqref{eq-f0-fin} into \eqref{eq:f_age_fct_of_pdf_at_arrival}. To calculate this expression, we first show the calculation of a generic term that represents the core of this expression. We compute
\begin{align}
  & \int_0^x\int_{x-x_0}^{+\infty} x_0^k t_1^{\ell} e^{-d'x_0 -\tilde{d}t_1} dt_1 dx_0 \nonumber \\
   & = \int_0^x x_0^k e^{-d'x_0}  \left[-e^{-\tilde{d}t_1} \sum_{i=0}^{\ell}\frac{\ell!}{i!\tilde{d}^{l-i+1}}t_1^i\right]_{x-x_0}^{\infty} dx_0 \nonumber \\
   & = \int_0^x x_0^k e^{-d'x_0}  \left(e^{-\tilde{d}(x-x_0)} \sum_{i=0}^{\ell}\frac{\ell!}{i!\tilde{d}^{l-i+1}}(x-x_0)^i\right) dx_0 . \label{eq:generic_term_eq_f}
\end{align}
The expression in \eqref{eq:generic_term_eq_f} stems from the fact that a primitive of
$e^{-\tilde{d}t_1}P(t_1)$, with polynomial $P(t_1)=t_1^{\ell}$, is $-e^{-\tilde{d}t_1}\sum_{i=0}^{\mathrm{deg}(P)} \frac{P^{(i)}(t_1)}{\tilde{d}^{i+1}}$ where $P^{(i)}$ is the $i$th derivative of $P$ and $P^{(0)}=P$. This sum can be written in a compact form as $\sum_{i=0}^{\ell}\frac{\ell!}{i!\tilde{d}^{l-i+1}}t_1^i$. For the evaluation of the integral we used that $\lim_{t_1\rightarrow \infty} t_1^{\ell}e^{-\tilde{d}t_1} = 0$ for any fixed $\ell$ and positive $\tilde{d}$.

Using the binomial theorem to expand the term $(x-x_0)^i$ in the expression above we can rewrite \eqref{eq:generic_term_eq_f} as
\begin{align}
   & \frac{\ell! e^{-\tilde{d}x}}{\tilde{d}^{\ell+1}}  \int_0^x x_0^k e^{-(d'-\tilde{d})x_0}  \sum_{i=0}^{\ell}\frac{ \tilde{d}^{i}}{i!}(x-x_0)^i dx_0 \nonumber\\
   & = \frac{\ell! e^{-\tilde{d}x}}{\tilde{d}^{\ell+1}}  \int_0^x x_0^k e^{-(d'-\tilde{d})x_0}  \sum_{i=0}^{\ell} c_{i,\ell}(x) x_0^i dx_0 . \label{eq:generic_term_eq_f_calc}
\end{align}
where we expanded $(x-x_0)^i = \sum_{k=0}^{i}(-1)^k \binom{i}{k}x^{i-k}x_0^k$. Then we rearrange the sum terms in the first line to express it as a polynomial in $x_0$ with coefficients
\begin{equation}
  c_{i,\ell}(x) = \sum_{j=i}^{l} (-1)^i \binom{j}{i} x^{j-i} \frac{\ell!}{j!}
  \label{eq:coefficient_cell}
\end{equation}
with $j\geq i$. Here, we explicitly express the dependency of the coefficients on $x$ through $c_{i,\ell}(x)$.

In a last step to compute \eqref{eq:generic_term_eq_f_calc} we calculate for $d'\neq\tilde{d}$
\begin{align}
   & \frac{\ell! e^{-\tilde{d}x}}{\tilde{d}^{\ell+1}} \sum_{i=0}^{\ell} c_{i,\ell}(x)  \int_0^x x_0^{k+i} e^{-(d'-\tilde{d})x_0} dx_0 \nonumber\\
   & = \frac{\ell! e^{-\tilde{d}x}}{\tilde{d}^{\ell+1}} \sum_{i=0}^{\ell} c_{i,\ell}(x)  \left[-e^{-(d'-\tilde{d})x_0}\sum_{j=0}^{k+i}\frac{(k+i)!x_0^j}{j!(d'-\tilde{d})^{k+i-j+1}}\right]_{0}^{x} \nonumber\\
   & = \frac{\ell! e^{-\tilde{d}x}}{\tilde{d}^{\ell+1}} \sum_{i=0}^{\ell} c_{i,\ell}(x)  \left[\frac{(k+i)!}{(d'-\tilde{d})^{k+i+1}} \right.\nonumber\\
   & \left. -e^{-(d'-\tilde{d})x}\sum_{j=0}^{k+i}\frac{(k+i)!x^j}{j!(d'-\tilde{d})^{k+i-j+1}}\right]. \label{eq:generic_term_eq_f_calc_final}
\end{align}
For the case when $d'=\tilde{d}$ we obtain as a solution of \eqref{eq:generic_term_eq_f_calc}
\begin{align}\label{eq:generic_term_eq_f_calc_final_d_equal}
  \frac{\ell! e^{-\tilde{d}x}}{\tilde{d}^{\ell+1}} \sum_{i=0}^{\ell} c_{i,\ell}(x) \frac{x^{k+i+1}}{k+i+1}
\end{align}
%

Now, using the steps from above we insert the formulation of the PDF \eqref{eq-f0-fin} into \eqref{eq:f_age_fct_of_pdf_at_arrival} and obtain the PDF of the age at any point in time in the following theorem.
\begin{theorem}
In a stationary $M/M/I_{\max}/I_{\max}^*$  system, the PDF of the age of information at an arbitrary point in time, $f(x)$, is given by
\begin{align}
&f(x) = \hat{\lambda} \sum_{i=0}^{I_{\max}}\sum_{j=0}^{I_{\max}} \sum_{(n',n) \mst 1\leq n+1\leq n'\leq I_{\max}} p^\circ_{n',n} \,e^{-x\tilde{d}_j}  \, \times\nonumber\\
  &\mkern-38mu\sum_{l=0}^{\tilde{\xi}_{j,n}+\xi_{i,n',n}} \mkern-28mu \tilde{c}_l(x) \mkern-4mu \left(\mkern-4mu\frac{l!}{{(d'_i-\tilde{d}_j)}^{l+1}} - e^{-x (d'_i-\tilde{d}_j)}\mkern-4mu\sum_{v=0}^{l}\frac{l!}{v!{(d'_i-\tilde{d}_j)}^{l-v+1}}x^{v}\mkern-4mu\right)
    \label{eq:f_age_any_point_in_time_complete}
\end{align}
 where
 \begin{itemize}
   \item $\lambda$ is the message generation rate at the sender, $1/\mu$ is the mean message transit time and $I_{\max}$ is the maximum number of messages in transit;
   \item $p^\circ_{n',n}$ from \eqref{eq-p0}, $d'_i$ from \eqref{eq-dprime}, and $\tilde{d}_j$ from \eqref{eq-dtilde};
   \item $\xi_{i,n',n},\tilde{\xi}_{j,n}$ are the degrees of the polynomials ${\pi}^{i,n',n}(x_0)$ and $\tilde{\pi}^{j,n}(t_1)$ from \eqref{eq-g0-fin} and \eqref{eq:ht1_given_n}. Specifically, these are given as $$\pi^{i,n',n}(x_0):=\sum_{k=0}^{\xi_{i,n',n}} a_k^{i,n',n} x_0^k$$ and $$\tilde{\pi}^{j,n}(t_1):=\sum_{k=0}^{\tilde{\xi}_{j,n}} \tilde{a}_k^{j,n} t_1^k.$$
   \item $\tilde{c}_l(x)$ are the polynomial coefficients obtained through the convolution
   \begin{equation}\label{eq:cell_final}
     \tilde{c}_l(x) = \sum_{v=0}^{l} \tilde{z}_v^{j,n}(x) a_{l-v}^{i,n',n}
   \end{equation}
   with $\tilde{z}_k^{j,n}(x) = \sum_{i=k}^{\tilde{\xi}_{j,n}} \tilde{a}_i^{j,n} c_{k,i}(x)$, where $c_{k,i}(x)$ is given in \eqref{eq:coefficient_cell}.
 \end{itemize}
\end{theorem}
\begin{proof}
By inserting the formulation of the PDF \eqref{eq-f0-fin} into \eqref{eq:f_age_fct_of_pdf_at_arrival} we obtain the PDF of the age at any point in time as
\begin{align}
  &f(x) = \hat{\lambda} \int_{0}^{x}  \int_{x-x_0}^{\infty} f^\circ(x_0,t_1) dt_1 dx_0 \nonumber\\
  &=\hat{\lambda} \sum_{i=0}^{I_{\max}}\sum_{j=0}^{I_{\max}} \sum_{(n',n) \mst 1\leq n+1\leq n'\leq I_{\max}} p^\circ_{n',n} \, \times\nonumber\\
  &\int_{0}^{x} e^{-x_0d'_i} \pi^{i,n',n}(x_0) e^{-(x-x_0)\tilde{d}_j} \, \times \nonumber\\
  &\int_{x-x_0}^{\infty} \tilde{\pi}^{j,n}(t_1) e^{-t_1\tilde{d}_j}  dt_1    dx_0
  \nonumber\\
  &=\hat{\lambda} \sum_{i=0}^{I_{\max}}\sum_{j=0}^{I_{\max}} \sum_{(n',n) \mst 1\leq n+1\leq n'\leq I_{\max}} p^\circ_{n',n} \, \times\nonumber\\
  &\int_{0}^{x} e^{-x_0d'_i} \sum_{k=0}^{\xi_{i,n',n}} a_k^{i,n',n} x_0^k e^{-(x-x_0)\tilde{d}_j} \, \times \nonumber\\
  &\int_{x-x_0}^{\infty} \sum_{k=0}^{\tilde{\xi}_{j,n}} \tilde{a}_k^{j,n} t_1^k e^{-t_1\tilde{d}_j}  dt_1    dx_0
  \label{eq:intermediate_1_pi_substituted}
\end{align}
Looking closely at \eqref{eq:intermediate_1_pi_substituted} after rearranging terms and swapping the sum in $\tilde{\pi}^{j,n}(t_1)$ and the integral over $t_1$ we observe that at the core of the problem we need to compute the expression in \eqref{eq:generic_term_eq_f}. The additional complexity compared to the result in \eqref{eq:generic_term_eq_f_calc_final} arises due to the sums in $\tilde{\pi}^{j,n}(t_1)$ and $\pi^{i,n',n}(x_0)$.
In the next step, we evaluate the second integral in \eqref{eq:intermediate_1_pi_substituted} using the same method as for \eqref{eq:generic_term_eq_f} to obtain
\begin{align}
  &\hat{\lambda} \sum_{i=0}^{I_{\max}}\sum_{j=0}^{I_{\max}} \sum_{(n',n) \mst 1\leq n+1\leq n'\leq I_{\max}} p^\circ_{n',n} \, \times\nonumber\\
  &\int_{0}^{x} e^{-x_0d'_i} \sum_{k=0}^{\xi_{i,n',n}} a_k^{i,n',n} x_0^k \sum_{k=0}^{\tilde{\xi}_{j,n}}\tilde{a}_k^{j,n} \sum_{v=0}^{k}\frac{k!}{v!\tilde{d}_j^{k-v+1}}(x-x_0)^{v} dx_0\nonumber\\
  &=\hat{\lambda} \sum_{i=0}^{I_{\max}}\sum_{j=0}^{I_{\max}} \sum_{(n',n) \mst 1\leq n+1\leq n'\leq I_{\max}} p^\circ_{n',n} \, \times\nonumber\\
  &\int_{0}^{x} e^{-x_0d'_i} \sum_{k=0}^{\xi_{i,n',n}} a_k^{i,n',n} x_0^k \sum_{k=0}^{\tilde{\xi}_{j,n}}\tilde{z}_k^{j,n}(x) x_0^k dx_0 .
  \label{eq:intermediate_2_pi_substituted}
\end{align}
Here, in the second line we expanded $(x-x_0)^{v}$ in the same way as in \eqref{eq:generic_term_eq_f_calc}. The difference to \eqref{eq:generic_term_eq_f_calc} stems from the additional sum leading to the intermediate form $\sum_{k=0}^{\tilde{\xi}_{j,n}}\tilde{a}_k^{j,n} \sum_{i=0}^{k} c_{i,k}(x) x_0^i$ after using the expansion in \eqref{eq:generic_term_eq_f_calc}.
Now, after collecting the terms we can rewrite this sum as $\sum_{k=0}^{\tilde{\xi}_{j,n}}\tilde{z}_k^{j,n}(x) x_0^k$ with $\tilde{z}_k^{j,n}(x) = \sum_{i=k}^{\tilde{\xi}_{j,n}} \tilde{a}_i^{j,n} c_{k,i}(x)$, where $c_{k,i}(x)$ is given in \eqref{eq:coefficient_cell}.

Inspecting \eqref{eq:intermediate_2_pi_substituted}, we see that the product of the two given polynomials in $x_0$ can be rewritten as one polynomial $\sum_{l=0}^{\tilde{\xi}_{j,n}+\xi_{i,n',n}} \tilde{c}_l(x) x_0^l$ with $\tilde{c}_l(x) = \sum_{v=0}^{l}\tilde{z}_v^{j,n}(x)a_{l-v}^{i,n',n}$, i.e., the convolution of the coefficients of the two polynomials.
Equipped with the integral evaluation in \eqref{eq:generic_term_eq_f_calc_final} we can now compute
\begin{align}
  &\hat{\lambda} \sum_{i=0}^{I_{\max}}\sum_{j=0}^{I_{\max}} \sum_{(n',n) \mst 1\leq n+1\leq n'\leq I_{\max}} p^\circ_{n',n} \, \times\nonumber\\
  &\int_{0}^{x} e^{-x_0d'_i} \sum_{l=0}^{\xi_{i,n',n}+\tilde{\xi}_{j,n}} \tilde{c}_l(x) x_0^l dx_0\nonumber\\
  &=\hat{\lambda} \sum_{i=0}^{I_{\max}}\sum_{j=0}^{I_{\max}} \sum_{(n',n) \mst 1\leq n+1\leq n'\leq I_{\max}} p^\circ_{n',n} \,e^{-x\tilde{d}_j}  \, \times\nonumber\\
&\mkern-38mu\sum_{l=0}^{\tilde{\xi}_{j,n}+\xi_{i,n',n}} \mkern-28mu \tilde{c}_l(x) \mkern-4mu \left(\mkern-4mu\frac{l!}{{(d'_i-\tilde{d}_j)}^{l+1}} - e^{-x (d'_i-\tilde{d}_j)}\mkern-4mu\sum_{v=0}^{l}\frac{l!}{v!{(d'_i-\tilde{d}_j)}^{l-v+1}}x^{v}\mkern-4mu\right)\nonumber\\
\end{align}

\end{proof}

\subsection{Age Distribution at the Arrival of Informative Messages}

In addition to calculating the age distribution at any point in time we can easily calculate the distribution of the age at the \emph{arrival instants of informative messages}. The corresponding PDF $f_A^\circ(x_0)$ is readily obtained as

\begin{equation*}
  f_A^\circ(x_0)=\sum_{(n',n) \mst 1\leq n+1\leq n'\leq I_{\max}} p^\circ_{n',n} g^\circ(x_0|n',n)
\end{equation*}
Using \eqref{eq-g0-fin} we obtain
\begin{align}
  f_A^\circ(x_0)=\sum_{i=0}^{I_{\max}}e^{-x_0d'_i} \mkern-18mu
  \sum_{(n',n) \mst 1\leq n+1\leq n'\leq I_{\max}} p^\circ_{n',n} \pi^{i,n',n}(x_0)
 \label{eq-ff0-fin}
\end{align}


\begin{figure*}
\centering
  \begin{subfigure}[b]{0.49\linewidth}
    \includegraphics[width=\linewidth]{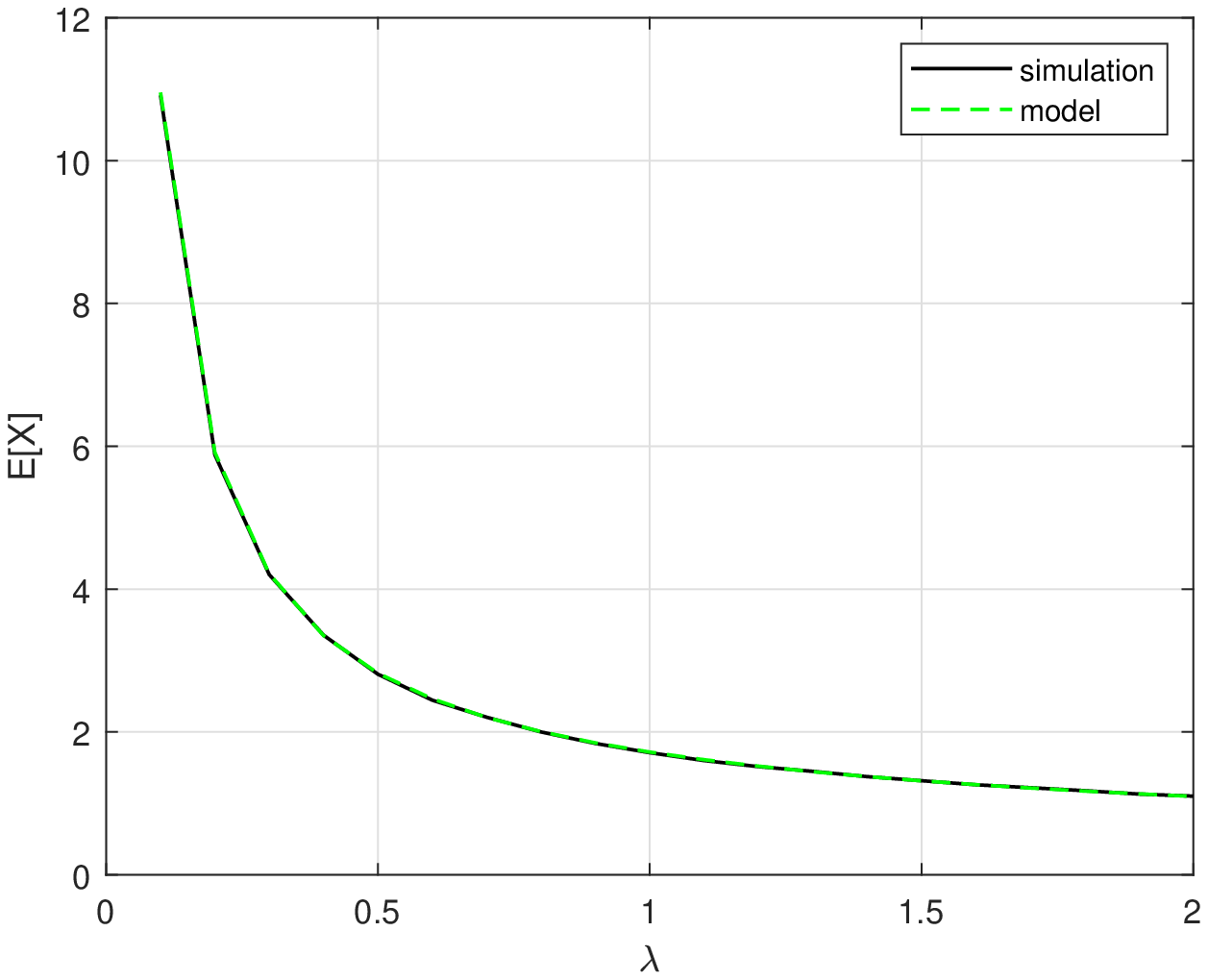}
    \caption{}
    \label{fig:mean_age_anytime_I_20}
  \end{subfigure}
  \hfill
    \begin{subfigure}[b]{0.49\linewidth}
    \includegraphics[width=\linewidth]{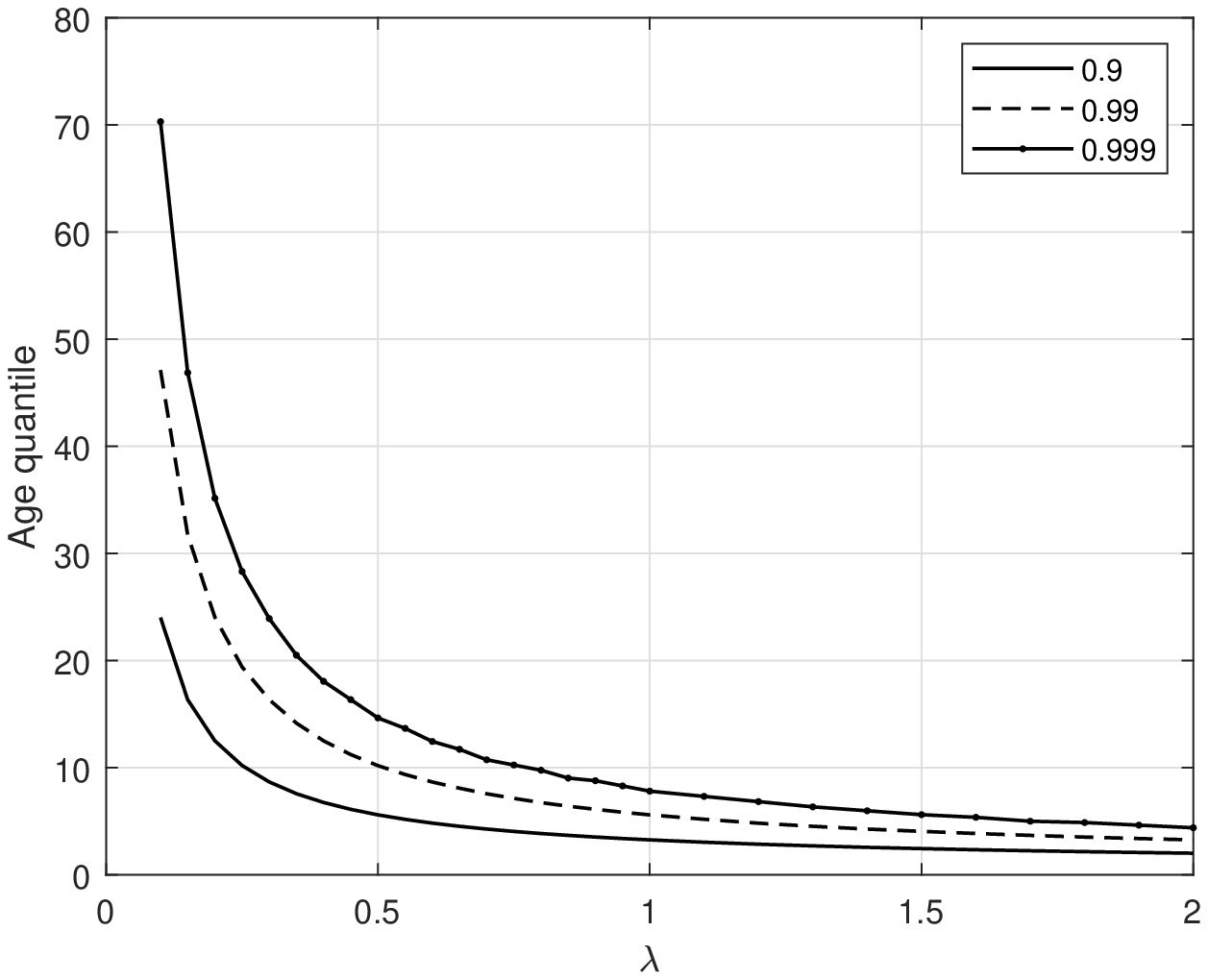}
    \caption{}
    \label{fig:quantile_age_anytime_I_20}
  \end{subfigure}
 \caption{(a)  Expected age $E[X]$ obtained from simulations compared to the model in \eqref{eq:expectation_test_fct_age} with $\varphi$ being the identity function for $I_{\max}=20$ and $\mu=1$. (b) Quantiles of the age $P[X>x_{\varepsilon}]=\varepsilon$ obtained from integrating the age density in \eqref{eq:f_age_fct_of_pdf_at_arrival} for $I_{\max}=20$ and $\mu=1$.}
\end{figure*}
%
%
%
%

\vspace{-10pt}
\section{Evaluation}
\label{sec:evaluation}

In this section, we compare the obtained expressions for the age distributions to results from empirical discrete event simulations.
We consider the system as described in Sect.~\ref{sec:palm_model} with messages arriving as a Poisson process with rate $\lambda$ where each message observes a service time sampled from an exponential distribution with parameter $\mu$.
The system simulation results are obtained from simulation runs over $10^5$ messages and we set the number of non-obsolete messages under way to $I_{\max}=20$ .

Figure~\ref{fig:ccdf_age_at_inform_2} shows the age distribution
at the arrival time points of informative messages. The dashed curve is obtained from the model \eqref{eq:expectation_test_fct_age} (with test function $\varphi$ being the identity function). The figure also shows the empirical age distribution obtained from simulations. We observe that these two distributions match very well and the impact of the average service rate of one message $\mu$ on the tail of the age distribution.
Fig.~\ref{fig:pdf_age_at_anytime} shows the probability density of the age at any point in time that is obtained from \eqref{eq:f_age_fct_of_pdf_at_arrival} using the Laplace inverse of \eqref{eq:joint_LST_x_t}. Observe the skewness of the density function.
This shows that approximations based on the first few (two) moments, e.g. obtained based on work that calculate the moments of the age distribution~\cite{Yates:MDS:TINT20}, will be inaccurate.

Figure~\ref{fig:mean_age_anytime_I_20} shows a comparison of the expected age at any point in time as a function of the message arrival rate $\lambda$.
The expected age that is obtained from the model is computed using the density in \eqref{eq:f_age_fct_of_pdf_at_arrival} in closed form.

To empirically obtain the average age from the event based simulation we utilize the Palm inversion formula \eqref{eq:palm_inversion_formula} with $\varphi$ set as identity function. 
Hence we can write
    \begin{align}
    &\E\left[X_t\right] = \mu \bar{N} \E^\circ\left[\int_{T_0}^{T_1}X_s ds \right] = \mu \bar{N} \E^\circ\left[\int_{0}^{T_1}(A_0+s) ~ds \right] \nonumber \\
    &= \mu \bar{N} \E^\circ\left[A_0(T_1-T_0) + \frac{1}{2} (T_1-T_0)^2 \right]
    \label{eq:plam_expectation}
  \end{align}
where $A_0$ is the age of the informative message received at time $T_0$. The estimate of $\E\left[X_t\right]$ obtained from one simulation run is
\begin{equation}
      \mu \hat{N}
      \sum_{n=1}^{n_{\mathrm{tot}}-1}
      \left( A_n(T_{n+1}-T_n)+\frac12 (T_{n+1}-T_n)^2
      \right)
  \end{equation}
where $A_n$ [resp. $T_n$] is the age upon delivery at the receiver [resp. delivery time] of the $n$th informative message, $n_{\mathrm{tot}}$ is the total number of informative messages delivered in the simulation run, and $\hat{N}$ is the time-average number of messages in the channel.

Here too, the comparison with the empirically obtained average age from \eqref{eq:plam_expectation} shows a close match. Note that the empirically obtained average age still requires invoking the Palm inversion formula \eqref{eq:palm_inversion_formula} as given in \eqref{eq:plam_expectation}.
We observe in Fig.~\ref{fig:mean_age_anytime_I_20} that the expected age  decreases monotonically with the message arrival rate $\lambda$. This stand in line with similar age models with finite message capacity assuming, however, FIFO message delivery, such as in \cite{Kam16}.
In Fig.~\ref{fig:quantile_age_anytime_I_20} we show the quantiles of the age at any point in time based on the age probability density in \eqref{eq:f_age_fct_of_pdf_at_arrival}. These quantiles can be utilized for system dimensioning by providing operating points, in terms of setting the service rate $\mu$ or throttling the message generation rate $\lambda$, to retain a corresponding age quantile $x_{\varepsilon}$ that is only violated with probability $P[X>x_{\varepsilon}]=\varepsilon$.

\section{Related Work}

The problem of status updating to combat data staleness in distributed systems that use a shared and unreliable network was first discussed in the context of  real-time database systems in \cite{songliu}.
Essentially, a recent reincarnation of this problem in the context of IoT that is known as Age of information (AoI) considers transmission scheduling strategies to update the status at some receiver in a way that optimizes the freshness of that information~\cite{Kaul12,sun18,ioannidis09}. This problem has been in particular of interest in the context of vehicular networks~\cite{kaul11,kaul11b} and  sensory information transmission over wireless networks~\cite{He:18,Hribar:17} as the freshness of information such as the captured environment model that is exchanged between vehicles is safety critical. For a comprehensive review see~\cite{YatesSBKMU21a}.


Given a single sender and a system modeled as an M/M/1 queue the work in \cite{Kaul12} derives the \emph{sample path average age} at the receiver as $\lim_{T\rightarrow\infty} \frac{1}{T}\int_{0}^{T}\Delta(t) dt = \lambda\left(E[XS]+E[X^2]/2\right)$ with $\Delta(t)$ denoting the age of information at the receiver at time $t$, $\lambda$ being the message arrival rate and the random variables $X,S$ that denote the message inter-arrival time and message response time, respectively.
The same seminal work provides expressions for the sample path average age in a D/M/1 system using a transcendental function.
Given the forms derived above the work \cite{Kaul12} also provides the parametrization that minimizes the average age.

Similarly, works such as \cite{Kam13,Kam16:ToIT} consider the age at the receiver given a $D/G/1$ and $M/M/\infty$ systems where messages may arrive out of order.
The work in \cite{Kam13} considers the distribution of the age process for a deterministic transmission schedule and a single server with general service time distribution under the FIFO assumption.
The authors derive in \cite{Kam16:ToIT} an expression for the \emph{average age} using a similar reasoning as \cite{Kaul12} as a function of the average message arrival rate and service time distribution.
Note that the provided expression there is not directly computable as it contains multiple infinite sums and infinite products.
Applications of the methods above to the special case of G/G/1/1 queueing systems under message blocking and message preemption is found in \cite{soysal21}.
In contrast to this work, however, our approach here provides the distribution of the age using a computable closed form that is constructed as the combination of Laplace-Stieltjes transforms of elementary functions

Going beyond elementary queues, the work in \cite{Talak2017MinimizingAI} considers the AoI for a path consisting of a concatenation of multiple links where the random delay at each of the links is only due random access. The authors show that given this delay model and a graph model of the network the problem of finding a transmission strategy for minimizing the AoI at $N$ sender-receiver pairs can be decomposed into a simpler equivalent optimization problem. We believe that the main reason for this lies in the delay due to random access model that does not incorporate queueing and scheduling effects.
A similar work considering multihop networks, i.e., \cite{Bedewy17}, that considers, however, a multihop queueing network, shows that a preemptive  Last Generated First Served (LGFS) policy at all nodes minimizes any non-decreasing functional of the age in stochastic dominance sense. This result is obtained under the assumption that all message transmission times are iid exponentially distributed at all nodes.

The works in \cite{Yates:multiple_source:TINT19,Yates:MDS:TINT20} provide a method to calculate the MGF and the moments of the AoI at a network monitor for networks of \emph{preemptive finite buffer servers} based on the so called stochastic hybrid system (SHS) framework~\cite{Hu:SHS:2000} by leveraging a notion of a hybrid state $[\{q(t)\}_{t\geq 0},\mathbf{x}(t)]$ where $\{q(t)\}_{t\geq 0}$ describes a continuous time Markov chain over finite state and $\mathbf{x}(t)\geq 0 \in \mathbb{R}^{1\times n}$ describes a vector of positive real values of the age.
By attaching deterministic matrices $\mathbf{A_l}\in\{0,1\}^{n\times n}$ to the transitions $l\in L$ of the Markov chain $\{q(t)\}_{t\geq 0}$, where $L$ denotes the set transitions, one is able to track the jumps of the age vector as $\mathbf{x'}=\mathbf{x}\mathbf{A_l}$.
The key to finding a formulation for the expected age and for the Moment generating function (MGF) of the age is based on a set of first order differential equations that assume $\{q(t)\}_{t\geq 0}$ is ergodic and utilize its steady state stationary probability distribution~\cite{Yates:multiple_source:TINT19,Yates:MDS:TINT20}
Note that the exists a direct relation between the presented SHS framework and utilizing the Master equation $\frac{d}{d t} E\left[\varphi(q_t,X_t)\right] = E\left[G \varphi(q_t,X_t)\right]$ with $\varphi$ being a test function, $G$ a generator and $q_t$ and $X_t$ describing the queue state, as well as, the age at time $t$ respectively.
This relation is explored in~\cite{Yates:multiple_source:TINT19,Yates:MDS:TINT20} to simplify the SHS formulation.
We note that, in general, applying the master equation to the infinitesimal generator that describes the jump and drift evolution of the age results in the Fokker-Plank equation describing the time evolution of the age density.
Analytical closed-form results to solve this formulation, even for the stationary age distribution, are yet to be shown.
Computable solutions to the presented SHS system of equations in are provided for the examples of a single M/M/1/1 queue, a line network of M/M/1/1 queues~in\cite{Yates:MDS:TINT20}.
Note that the SHS framework was applied in~\cite{Yates18} to obtain a close form for the expected age for a system of parallel servers where a new message arrival preempts the oldest message under way.

Concerning message reordering, a seminal work on packet reordering is \cite{BaccelliGP84} that provides a recursive expression for the total delay distribution of a system in which messages that arrive in order are delivered through a disordering system, hence the out-of-order arrivals require resequencing. The author provides an analytical solution for the case of an $M/G/\infty$ disordering system given in terms of a Laplace-Stieltjes transform.

A min-plus approach to the Age of Information is given in~\cite{Fidler21:AoI} showing that the virtual delay at a FIFO system, i.e., the horizontal deviation of cumulative arrivals and departures at a min-plus system, is an upper bound on the age. Equipped with lower bounds on the cumulative arrival traffic, the work in~\cite{Fidler21:AoI} shows deterministic and statistical upper bound on the age of information for different combinations of arrivals and systems with deterministic or probabilistic description.
In contrast to~\cite{Fidler21:AoI}, we consider here non-FIFO systems with possible message reordering.

\section{Conclusion}
\label{sec:conclusion} 

In this paper, we considered the problem of computing the distribution of the Age of Information (AoI) at any point in time for non-preemptive, non-FIFO systems.
Our key observation is that this networked system can be modeled as a batch queueing system where the served batch size is random and the sojourn time of the freshest message in the batch corresponds to the age of an arriving informative message. 
The batch (except for the its freshest message) essentially models the set of messages that are generated before the freshest message and hence are rendered obsolete by its arrival. 
This captures message reordering due to the non-FIFO system property.

Equipped with this queueing model we use Palm calculus together with time inversion to decompose the elements that form the joint distribution of the age and the time between the arrival of informative messages at the receiver.
Then, Palm inversion allows us to compute the distribution of the age at any point in time given this joint distribution.
We find recursions for the corresponding Laplace-Stieltjes transforms of the conditional age and informative message inter-arrival time distributions owing to the Markovian nature of the underlying model.
As these transforms turn to be rational we obtain a computable expression for the AoI distribution composed of matrix-exponential terms. 
This main result allows further to deduce formulations for the expected age, as well as, the age distribution at the arrival time points of informative messages.
We validate the exact model using discrete-event simulations and show the skewness of the PDFs of the age.
Further, we show the impact of the arrival and service rates on the age CDF and its quantiles.
We leave the extension of this model to multi-stage queueing networks to future work. 

\section{Appendix}
\label{sec:Appendix}

\subsection{Calculation of the steady state probabilities in \eqref{eq:steady_state_prob_chain_nr_msgs_underway}}
\label{sec:appendix_proof_steady_state_prob_fwd}

We consider the queueing model from Sect.~\ref{sec:underlying_model} with  the corresponding Markov chain as sketched in Fig.~\ref{fig:MC_fwd}. In the following we prove the formulation of the steady state probabilities $p_n$ of the given Markov chain through induction.
We first show that the formulation holds for $p_0$ and $p_1$. Then we show that given that the formulation holds for all $p_k$ for $k<n$ it also holds for $p_n$.

From the balance equations we can write $p_o\lambda = \sum_{i=1}^{I_{\max}}p_i\mu$. From the normalization condition $\sum_{i=0}^{I_{\max}}p_i=1$ we obtain $p = \frac{\mu}{\lambda+\mu}$ as $\sum_{i=1}^{I_{\max}}p_i =1-p_0$.
For $p_1$ we can write $p_1(\lambda+\mu) = p_0\lambda + \sum_{i=2}^{I_{\max}}p_i\mu$. Using $p_0$ and the normalization condition this reduces to $p_1\lambda = p_0 \lambda + \mu(1-p_0-p_1)-p_1\mu$ which leads to $p_1 = \frac{2\lambda\mu}{(\lambda+\mu)(\lambda+2\mu)}$.

Now, considering \eqref{eq:steady_state_prob_chain_nr_msgs_underway} we directly see that it holds for $n=0$ and $n=1$.
For a state $k$ of the given Markov chain we can write using the balance equations
\begin{align}
p_k(\lambda+k\mu) &= p_{k-1}\lambda+\mu\sum_{i=k+1}^{I_{\max}}p_i \nonumber\\
&= p_{k-1}\lambda+\mu(1- \sum_{i=0}^{k}p_i) \,,
\label{eq:steady_state_prob_fwd_general_recursion}
\end{align}
which we can rewrite as
\begin{align}
&p_k(\lambda+(k+1)\mu) \nonumber\\
&= p_{k-1}(\lambda-\mu)+\mu - \mu \sum_{i=0}^{k-2}p_i \nonumber\\
&= \frac{k\lambda^{k-1}\mu(\lambda-\mu)}{\prod_{j=1}^{k}(\lambda+j\mu)} + \mu - \mu \sum_{i=0}^{k-2}\frac{(i+1)\lambda^i\mu}{\prod_{j=1}^{i+1}(\lambda+j\mu)}\nonumber\\
&=\frac{k\lambda^{k-1}\mu(\lambda-\mu)\Gamma\left(\frac{\lambda+\mu}{\mu}\right)}{\mu^k\Gamma\left(\frac{\lambda+(k+1)\mu}{\mu}\right)} + \mu\nonumber\\
& - \mu\Gamma\left(\frac{\lambda+\mu}{\mu}\right)\left(\frac{\lambda+\mu}{\mu\Gamma\left(\frac{\lambda+2\mu}{\mu}\right)} - \frac{\lambda^{k-1}(k\mu+\lambda)}{\mu^k\Gamma\left(\frac{\lambda+(k+1)\mu}{\mu}\right)}\right) \,,
\label{eq:steady_state_prob_fwd_general_recursion2}
\end{align}
where we used the identity $\prod_{j=1}^{k}(\lambda+j\mu) = \frac{\mu^k\Gamma\left(\frac{\lambda+(k+1)\mu}{\mu}\right)}{\Gamma\left(\frac{\lambda+\mu}{\mu}\right)}$.
Now, we can further simplify \eqref{eq:steady_state_prob_fwd_general_recursion2} using an instance of this the identity $\Gamma\left(\frac{\lambda+2\mu}{\mu}\right) = \frac{\lambda+\mu}{\mu} \Gamma\left(\frac{\lambda+\mu}{\mu}\right)$.
Finally through rearranging terms we obtain
\begin{align}
p_k =  \frac{(k+1) \lambda^k\mu \Gamma\left(\frac{\lambda+\mu}{\mu}\right)}{\mu^k\Gamma\left(\frac{\lambda+(k+1)\mu}{\mu}\right)} \,,
\label{eq:steady_state_prob_fwd_final_result}
\end{align}
which completes the proof.
Calculating $p_k$ for $k = I_{\max}$ follows along using the balance equation as shown above.

\subsection{Calculation of the transition rates of the reversed process \eqref{eq:lambda_prime}, \eqref{eq:mu_prime}}

The transition rates for the reversed Markov process are obtained directly from [Theorem 1.12] from \cite{Kelly:Reversibility-2011} as $Q'(i,j) = \frac{p_j Q(j,i)}{p_i}, i,j \in E$ with the same steady state distribution $p_n n\in E$.
Now given the forward Markov process with transition rates in \eqref{eq:transition_rates_fwd_process} (as sketched in Fig.~\ref{fig:MC_fwd}) we obtain the following transition rates for the reverse process $Q'_{i,i-1} :=\lambda_{i}'$ for $i=1...I_{\max}$, $Q'_{i,j} :=\mu_{ij}'$ for $i=0...I_{\max}-1,i<j\leq I_{\max}$ and $Q'_{i,j} :=0$ otherwise.
Hence, we obtain from \eqref{eq:transition_rates_fwd_process} and \eqref{eq:steady_state_prob_chain_nr_msgs_underway} for $i<I_{\max}$
\begin{align}\label{eq:derivation_lambda_prime}
  \lambda_i' &= \lambda\frac{p_{i-1}}{p_{i}} = \lambda \frac{i\lambda^{i-1}\mu}{\prod_{j=1}^{i}\left(\lambda+j\mu\right)} \frac{\prod_{j=1}^{i+1}\left(\lambda+j\mu\right)}{\left(i+1\right)\lambda^i\mu}\nonumber\\
  &= \frac{i}{i+1}\left(\lambda+(i+1)\mu\right) \,.
\end{align}
For $i=I_{\max}$ the derivation goes accordingly to find $\lambda_i' = I_{\max}\mu$.
Similarly, for $i<j<I_{\max}$ we obtain
\begin{align}\label{eq:derivation_mu_prime}
  \mu_{ij}' &= \mu\frac{p_{j}}{p_{i}} = \mu \frac{(j+1)\lambda^{j}\mu}{\prod_{k=1}^{j+1}\left(\lambda+k\mu\right)} \frac{\prod_{k=1}^{i+1}\left(\lambda+k\mu\right)}{\left(i+1\right)\lambda^i\mu}\nonumber\\
  &= \frac{(j+1) \mu \lambda^{j-i}}{(i+1)\prod\limits_{k=i+2}^{j+1}\left(\lambda + k\mu\right)} \,.
\end{align}
Again, for $j=I_{\max}$ the derivation goes similarly.

\subsection{On the numerical calculation of the LSTs in \eqref{eq:joint_LST_x_t}}
\label{app-inv}

For completeness, we show in the following an alternative method to the direct calculation of the conditional LST in \eqref{eq:joint_LST_x_t} that we used in this paper.
We underline that the following numerical computation may be beneficial in speeding up computations especially for evaluation purposes.

The direct computation of \eqref{eq:joint_LST_x_t} by computing the conditional LST vectors \eqref{eq:LST_of_kth_departure_given_state_n_rest_of_recursion} through recursion and matrix inversions, as well as, the LST vector in \eqref{eq:LST_of_next_arrival_given_state_nprime_matrix} and the following insertion of the vector components $f_{n,k}$ and $\tilde{f}_{n}$ into \eqref{eq:joint_LST_x_t} becomes computationally intensive when $I_{\max}$ is large.
The reason for this is the computation of the matrix inverse $\mt{\Phi^{-1}}$ in \eqref{eq:LST_of_1st_departure_given_state_n_initial_condition} as well its exponentiation in form of $\mt{\Psi^{k-1}}$ in \eqref{eq:LST_of_kth_departure_given_state_n_rest_of_recursion}.
Next we discuss an alternative numerical method to compute the quantities in \eqref{eq:LST_of_1st_departure_given_state_n_initial_condition} - \eqref{eq:LST_of_kth_departure_given_state_n_rest_of_recursion}.

First, we recognize that $\mt{\Phi^{-1} = (\theta I + D - M)^{-1}}$ used in \eqref{eq:LST_of_1st_departure_given_state_n_initial_condition} is a fraction by the adjugate matrix formula, i.e.,
\begin{equation}\label{eq:invphi_fraction}
  \mt{\Phi^{-1}} = \frac{\mt{R}(\theta)}{\rho(\theta)} \,,
\end{equation}
where $\mt{R}(\theta)$ is a polynomial with matrix coefficients given by
\begin{equation}\label{eq:R_polynomial}
  \mt{R}(\theta) = \sum_{k=0}^{I_{\max}}\theta^k\mt{P}_k \,,
\end{equation}
and the denominator is the characteristic polynomial of $\mt{M-D}$, i.e.
\begin{equation}\label{eq:denom_rho}
  \rho(\theta):=\sum_{k=0}^{I_{\max}+1}r_k\theta^k = \det\mt{(\Phi)}
\end{equation}
Also observe that $r_0 = \det\mt{(D-M)}$ and that $r_{I_{\max}+1} = 1$.
Note that the matrices $\mt{P}_k$ are of dimensions $(I_{\max}+1) \times (I_{\max}+1)$ large. Next, we obtain $\mt{P}_k$ iteratively using LeVerrier's method. In a nutshell, we plug \eqref{eq:invphi_fraction} into the identity $\mt{\Phi^{-1} \Phi = I}$ and rearrange the terms to obtain
\begin{equation}\label{eq:detthetaM}
  \sum_{k=1}^{I_{\max}+1} \theta^k \mt{P}_{k-1} + \sum_{k=0}^{I_{\max}}\theta^k \mt{P}_k \mt{(D-M)} =
  \left(\sum_{k=0}^{I_{\max}+1}r_k\theta^k\right)\mt{ I }  \,.
\end{equation}
Now, we can compare the coefficients of $\theta^k$ on both sides of \eqref{eq:detthetaM} and obtain the recursive form for the matrices $\mt{P}_k$
as
\begin{equation}\label{eq:Pkrecursive}
  \mt{P}_k= r_{k+1}\mt{I} - \mt{P}_{k+1}\mt{(D-M)} \,,
\end{equation}
for $k\in\{I_{\max-1},..,1\}$.
From the comparison of the coefficients in \eqref{eq:detthetaM} we know that $\mt{P}_0= \left(\det(\mt{D-M})\right)(\mt{D-M})^{-1}$ and $\mt{P}_{I_{\max}}=\mt{I}$ such that we can iteratively find the matrices $\mt{P}_k$ using \eqref{eq:Pkrecursive}, hence, calculate the coefficients of $\mt{R(\theta)}$.

Now, given that we calculate $\mt{\Phi^{-1}}$ using the method above we can use this result to simplify the matrix multiplication in $\mt{\Psi^{k}}$ as $\mt{\Psi= \Phi^{-1}\Lambda}$. Hence, we can write
\begin{equation}\label{eq:psipowerk}
  \mt{\Psi}^k = \frac{\tilde{\mt{R}}(\theta)^k}{\rho(\theta)^k} \,,
\end{equation}
where we used the polynomial $\tilde{\mt{R}}(\theta)$ that is defined as
\begin{equation}\label{eq:Rtildepolynomial}
  \tilde{\mt{R}}(\theta) = \sum_{k=0}^{I_{\max}}\theta^k\mt{\tilde{P}_k} \,,
\end{equation}
with $\mt{\tilde{P}_k = P_k \Lambda}$.
Now calculating the denominator of \eqref{eq:psipowerk} is simple as $\mt{\Phi}$ is triangular and its determinant is obtained in closed form as
\begin{equation}\label{eq:Rtildepolynomial}
  \det\mt{(\Phi)} = \prod_{i=0}^{I_{\max}} \left(\theta+\sum_{j}Q'_{i,j}\right) \,.
\end{equation}
As $\mt{\tilde{R}}(\theta)^k$ is a product of polynomials with matrix coefficients, we calculate the numerator in \eqref{eq:psipowerk} using an iterative convolution operation of the coefficients $\mt{\tilde{P}_k}$.
Now, given the calculation method above we can numerically obtain $\mt{\Psi^{k}}$ for insertion in \eqref{eq:LST_of_kth_departure_given_state_n_rest_of_recursion}.
Note that the same procedure can be used to obtain the elements $\tilde{f}_{n}(\nu)$ in \eqref{eq:joint_LST_x_t} by numerically calculating the inversion in \eqref{eq:LST_of_next_arrival_given_state_nprime_matrix}.

Finally, calculating the inverse Laplace transform of the Palm joint density $f^\circ(t_1,x_0)$ entails taking the inverse Laplace transform of the right hand side (RHS) of \eqref{eq:joint_LST_x_t}.
Given the factorization of $\rho(\theta)$ and the matrix coefficient form of the polynomial $\mt{R(\theta) = \sum_{k=0}^{I_{\max}}\theta^k\mt{P}_k}$
we observe that $f_{n',n+1}(\theta)$ on the RHS of \eqref{eq:joint_LST_x_t} has the form $\sum_i\frac{\alpha_i}{(\theta+d'_i)^{\kappa_i}}$ with constants $\alpha_i$ and $\kappa_i\leq n$ due to the partial fraction decomposition of \eqref{eq:psipowerk}.
To obtain the Palm joint density $f^\circ(t_1,x_0)$ we calculate the inverse Laplace-Stieltjes transform of the RHS of \eqref{eq:joint_LST_x_t}.
Given the observation that $f_{n',n+1}(\theta)$ can be rewritten as $\sum_i\frac{\alpha_i}{(\theta+d'_i)^{\kappa_i}}$ we know that the inverse LST of $f_{n',n+1}(\theta)$ has the form $\sum_i c_i x_0^{\kappa_i} e^{-d'_i x_0}$ with constants $c_i$ and $\kappa_i\leq n$.
The same observation holds for $\tilde{f}_{n}(\nu)$ in \eqref{eq:joint_LST_x_t}, i.e., by calculating the inversion of \eqref{eq:LST_of_next_arrival_given_state_nprime_matrix} using the method above we finally obtain a partial fraction decomposition and subsequent inverse LST that has the form $\sum_j h_j t_1^{\varsigma_j} e^{-\tilde{d}_j t_1}$.



%




\ifCLASSOPTIONcaptionsoff
  \newpage
\fi




\def\bibfont{\footnotesize}
\bibliographystyle{IEEEtran}
\bibliography{bibliography} 

\begin{thebibliography}{10}
\providecommand{\url}[1]{#1}
\csname url@samestyle\endcsname
\providecommand{\newblock}{\relax}
\providecommand{\bibinfo}[2]{#2}
\providecommand{\BIBentrySTDinterwordspacing}{\spaceskip=0pt\relax}
\providecommand{\BIBentryALTinterwordstretchfactor}{4}
\providecommand{\BIBentryALTinterwordspacing}{\spaceskip=\fontdimen2\font plus
\BIBentryALTinterwordstretchfactor\fontdimen3\font minus
  \fontdimen4\font\relax}
\providecommand{\BIBforeignlanguage}[2]{{%
\expandafter\ifx\csname l@#1\endcsname\relax
\typeout{** WARNING: IEEEtran.bst: No hyphenation pattern has been}%
\typeout{** loaded for the language `#1'. Using the pattern for}%
\typeout{** the default language instead.}%
\else
\language=\csname l@#1\endcsname
\fi
#2}}
\providecommand{\BIBdecl}{\relax}
\BIBdecl

\bibitem{Lee08:CPS}
E.~A. Lee, ``Cyber physical systems: Design challenges,'' in \emph{11th IEEE
  International Symposium on Object and Component-Oriented Real-Time
  Distributed Computing (ISORC)}, 2008, pp. 363--369.

\bibitem{Kaul12}
S.~Kaul, R.~Yates, and M.~Gruteser, ``Real-time status: How often should one
  update?'' in \emph{Proceedings IEEE INFOCOM}, 2012, pp. 2731--2735.

\bibitem{Yates:MDS:TINT20}
R.~D. Yates, ``The age of information in networks: Moments, distributions, and
  sampling,'' \emph{IEEE Transactions on Information Theory}, vol.~66, no.~9,
  pp. 5712--5728, 2020.

\bibitem{Bedewy17}
A.~M. Bedewy, Y.~Sun, and N.~B. Shroff, ``Age-optimal information updates in
  multihop networks,'' in \emph{2017 IEEE International Symposium on
  Information Theory (ISIT)}, 2017, pp. 576--580.

\bibitem{Popovski:19}
P.~Popovski, C.~Stefanovi\'c, J.~J. Nielsen, E.~de~Carvalho,
  M.~Angjelichinoski, K.~F. Trillingsgaard, and A.-S. Bana, ``Wireless access
  in ultra-reliable low-latency communication (urllc),'' \emph{IEEE
  Transactions on Communications}, vol.~67, no.~8, pp. 5783--5801, 2019.

\bibitem{Kam13}
C.~Kam, S.~Kompella, and A.~Ephremides, ``Age of information under random
  updates,'' in \emph{2013 IEEE International Symposium on Information Theory},
  2013, pp. 66--70.

\bibitem{Kam16:ToIT}
C.~Kam, S.~Kompella, G.~D. Nguyen, and A.~Ephremides, ``Effect of message
  transmission path diversity on status age,'' \emph{IEEE Transactions on
  Information Theory}, vol.~62, no.~3, pp. 1360--1374, 2016.

\bibitem{Fidler21:AoI}
M.~Noroozi and M.~Fidler, ``A min-plus model of age-of-information with
  worst-case and statistical bounds,'' \emph{CoRR}, vol. abs/2112.11934, 2021.

\bibitem{Yates:multiple_source:TINT19}
R.~D. Yates and S.~K. Kaul, ``The age of information: Real-time status updating
  by multiple sources,'' \emph{IEEE Transactions on Information Theory},
  vol.~65, no.~3, pp. 1807--1827, 2019.

\bibitem{YatesSBKMU21a}
R.~D. Yates, Y.~Sun, D.~R. Brown, S.~K. Kaul, E.~H. Modiano, and S.~Ulukus,
  ``Age of information: An introduction and survey,'' \emph{{IEEE} J. Sel.
  Areas Commun.}, vol.~39, no.~5, pp. 1183--1210, 2021.

\bibitem{Yates18}
R.~D. Yates, ``Status updates through networks of parallel servers,'' in
  \emph{2018 {IEEE} International Symposium on Information Theory, {ISIT} 2018,
  Vail, CO, USA, June 17-22, 2018}.\hskip 1em plus 0.5em minus 0.4em\relax
  {IEEE}, 2018, pp. 2281--2285.

\bibitem{baccelli2012palm}
F.~Baccelli and P.~Br{\'e}maud, \emph{Palm probabilities and stationary
  queues}.\hskip 1em plus 0.5em minus 0.4em\relax Springer Science \& Business
  Media, 2012, vol.~41.

\bibitem{serfozo2009basics}
R.~Serfozo, \emph{Basics of applied stochastic processes}.\hskip 1em plus 0.5em
  minus 0.4em\relax Springer Science \& Business Media, 2009.

\bibitem{le2010performance}
J.-Y. Le~Boudec, \emph{Performance evaluation of computer and communication
  systems}.\hskip 1em plus 0.5em minus 0.4em\relax Epfl Press Lausanne, 2010,
  vol.~2.

\bibitem{boudec2011performance}
J.~Le~Boudec, \emph{Performance Evaluation of Computer and Communication
  Systems}, ser. Computer and communication sciences.\hskip 1em plus 0.5em
  minus 0.4em\relax CRC Press, 2011.

\bibitem{gillespie1976general}
D.~T. Gillespie, ``A general method for numerically simulating the stochastic
  time evolution of coupled chemical reactions,'' \emph{Journal of
  computational physics}, vol.~22, no.~4, pp. 403--434, 1976.

\bibitem{Kelly:Reversibility-2011}
F.~P. Kelly, \emph{Reversibility and Stochastic Networks}.\hskip 1em plus 0.5em
  minus 0.4em\relax USA: Cambridge University Press, 2011.

\bibitem{Kam16}
C.~{Kam}, S.~{Kompella}, G.~D. {Nguyen}, and A.~{Ephremides}, ``Effect of
  message transmission path diversity on status age,'' \emph{IEEE Transactions
  on Information Theory}, vol.~62, no.~3, pp. 1360--1374, 2016.

\bibitem{songliu}
X.~Song and J.~Liu, ``Performance of multiversion concurrency control
  algorithms in maintaining temporal consistency,'' in \emph{Proceedings.,
  Fourteenth Annual International Computer Software and Applications
  Conference}, 1990, pp. 132--139.

\bibitem{sun18}
Y.~Sun, E.~Uysal-Biyikoglu, and S.~Kompella, ``Age-optimal updates of multiple
  information flows,'' in \emph{IEEE Conference on Computer Communications
  Workshops (INFOCOM WKSHPS)}, 2018, pp. 136--141.

\bibitem{ioannidis09}
S.~Ioannidis, A.~Chaintreau, and L.~Massoulie, ``Optimal and scalable
  distribution of content updates over a mobile social network,'' in \emph{IEEE
  INFOCOM 2009}, 2009, pp. 1422--1430.

\bibitem{kaul11}
S.~Kaul, M.~Gruteser, V.~Rai, and J.~Kenney, ``Minimizing age of information in
  vehicular networks,'' in \emph{8th Annual IEEE Communications Society
  Conference on Sensor, Mesh and Ad Hoc Communications and Networks}, 2011, pp.
  350--358.

\bibitem{kaul11b}
S.~Kaul, R.~Yates, and M.~Gruteser, ``On piggybacking in vehicular networks,''
  in \emph{IEEE Global Telecommunications Conference - GLOBECOM}, 2011, pp.
  1--5.

\bibitem{He:18}
Q.~He, G.~Dan, and V.~Fodor, ``Minimizing age of correlated information for
  wireless camera networks,'' in \emph{IEEE Conference on Computer
  Communications Workshops (INFOCOM WKSHPS)}, 2018, pp. 547--552.

\bibitem{Hribar:17}
J.~Hribar, M.~Costa, N.~Kaminski, and L.~A. DaSilva, ``Updating strategies in
  the internet of things by taking advantage of correlated sources,'' in
  \emph{IEEE Global Communications Conference}, 2017, pp. 1--6.

\bibitem{soysal21}
A.~Soysal and S.~Ulukus, ``Age of information in g/g/1/1 systems: Age
  expressions, bounds, special cases, and optimization,'' \emph{IEEE
  Transactions on Information Theory}, vol.~67, no.~11, pp. 7477--7489, 2021.

\bibitem{Talak2017MinimizingAI}
R.~Talak, S.~Karaman, and E.~H. Modiano, ``Minimizing age-of-information in
  multi-hop wireless networks,'' \emph{55th Annual Allerton Conference on
  Communication, Control, and Computing (Allerton)}, pp. 486--493, 2017.

\bibitem{Hu:SHS:2000}
J.~Hu, J.~Lygeros, and S.~Sastry, ``Towards a theory of stochastic hybrid
  systems,'' in \emph{Hybrid Systems: Computation and Control}.\hskip 1em plus
  0.5em minus 0.4em\relax Springer Berlin Heidelberg, 2000, pp. 160--173.

\bibitem{BaccelliGP84}
F.~Baccelli, E.~Gelenbe, and B.~Plateau, ``An end-to-end approach to the
  resequencing problem,'' \emph{J. {ACM}}, vol.~31, no.~3, pp. 474--485, 1984.

\end{thebibliography}

\balance

\end{document}